\documentclass[a4paper, amsfonts, amssymb, amsmath, reprint, showkeys, nofootinbib, twoside,notitlepage,onecolumn]{revtex4-1}

\usepackage{amsmath,amstext}
\usepackage[T1]{fontenc}
\usepackage{amssymb}
\usepackage{graphicx}
\usepackage{ae,aecompl}

\DeclareFontFamily{OT1}{pzc}{}
\DeclareFontShape{OT1}{pzc}{m}{it}{<-> s * [1.10] pzcmi7t}{}
\DeclareMathAlphabet{\mathpzc}{OT1}{pzc}{m}{it}

\usepackage{hyperref}
\usepackage{amsmath}
\usepackage{amssymb}
\usepackage{mathtools}
\usepackage{bm}
\usepackage{cleveref}
\usepackage{tensor}
\usepackage{braket}
\usepackage{enumitem}
\usepackage{mhchem}
\usepackage{amsthm}
\usepackage{nccmath}
\usepackage{mathrsfs}
\usepackage{color}

\newcommand{\E}{{\mathbb{E}}}
\newcommand{\K}{{\mathbb{K}}}
\newcommand{\G}{{\mathbb{G}}}

\newcommand{\s}{{\sigma \!\! \! \sigma}}

\def\be{\begin{equation}}
\def\ee{\end{equation}}
\def\beq{\begin{eqnarray}}
\def\eeq{\end{eqnarray}}

\theoremstyle{definition}

\theoremstyle{theorem}
\newtheorem{theorem}{Theorem}

\begin{document}
\title{Relativistic transport near moving interfaces}
\author{L.~Gavassino}
\affiliation{Department of Applied Mathematics and Theoretical Physics, University of Cambridge, Wilberforce Road, Cambridge CB3 0WA, United Kingdom}

\begin{abstract}
We study linear disturbances localized near planar surfaces moving at constant velocity $v$ in relativistic media. Depending on the physical setting, the surface may represent a moving obstacle, a thermal boundary, or an external source, providing a unified description of boundary layers, wakes, and the asymptotic tails of shock waves. The central result is a propagator representation of the interface solution that yields a geometric characterization of these phenomena. Using a Laplace-transform formulation, we show that the solution is a superposition of modes with purely imaginary frequency and wavenumber. For a given interface velocity, the admissible modes are selected by the line $i\omega=vik$ in the $\{i\omega,ik\}$ plane. As $v$ varies, this line sweeps across the spectrum, providing a unified geometric description of interface-localized solutions for arbitrary interface velocities. We illustrate the formalism with applications to relativistic hydrodynamics and kinetic theory.
\end{abstract} 
\maketitle

\section{Introduction}
\vspace{-0.2cm}

We study the following linear boundary-value problem, which we shall refer to as \textit{interface problem} (see Fig. \ref{fig:TheProblem}): 
\vspace{-0.1cm}

\begin{quote}
Let $\Psi(x^\mu)$ denote the dynamical variables of the system, linearized about a homogeneous background at rest. We seek disturbances of the form $\Psi(x-vt)$, defined in the half-space $x-vt\ge0$, that remain localized near the moving surface $x=vt$, in the sense that $\Psi(x-vt)$ grows at most polynomially as $x-vt\to+\infty$.
\end{quote}
\vspace{-0.1cm}

This problem arises naturally in a wide range of physical settings \cite{case1960elementary,zeldovich1967shock,mihalas_book,Prinja2010,schlichting2017boundary,NovakWithersObstacles:2018pnv,Kiselev2019}. The simplest example is the viscous boundary layer, where the stationary wall defines the interface $x=0$, and the tangential velocity grows linearly away from the wall. Another example is provided by a sharply localized electromagnetic or gravitational pulse propagating through a medium. Modeling the pulse as a source localized on the hypersurface $x=-t$, the wake left behind it is described by an interface problem with $v=-1$. Shock-wave asymptotics provide another natural realization of the same problem. Far from the shock front, the linear approximation becomes valid. Taking the hypersurface $x-vt=0$ to define the onset of the linear regime, the asymptotic tail is described by the corresponding linear interface problem.

Interface problems are particularly interesting in relativity, as the interface may move arbitrarily close to the speed of light, making relativistic kinematics and causality central to the structure of the solution. Despite this, systematic studies of how interface solutions depend on the interface velocity appear to be limited to the holographic analysis of \cite{NovakWithersObstacles:2018pnv}. The purpose of this work is to formulate and solve the interface problem in the context of transient hydrodynamics and kinetic theory. We will show that the interface problem admits a remarkably simple geometric characterization: the admissible exponential contributions are selected by the intersections of the linear excitation spectrum with the line $i\omega=vik$ in the real $\{i\omega,ik\}$ plane \cite{GavassinoLorentzianRelxation:2026seq}. The characteristic decay lengths are therefore determined entirely by this geometric construction, which arises naturally from a Laplace-transform representation of the governing equations.

\begin{figure}[h!]
    \centering
    \includegraphics[width=0.3\linewidth]{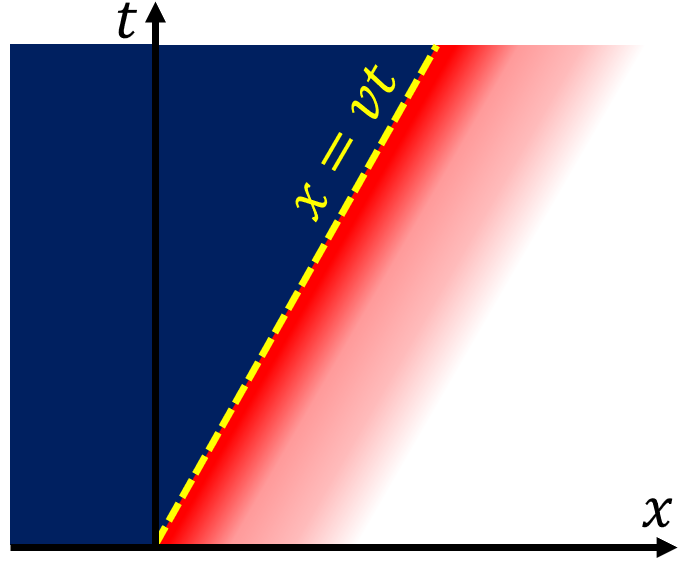}
\caption{Minkowski diagram illustrating the interface problem, shown in the rest frame of the background medium. The solution $\Psi(x-vt)$ (shades of red) is stationary in the frame comoving with the surface $x=vt$. It is defined only in the half-space $x\ge vt$ and is assumed to decay away from the surface (or, in some special cases, to grow polynomially). The complementary region $x<vt$ may contain a material boundary, a nonlinear continuation of the solution, or an external source. Depending on the sign of $v$, the solution describes a wake ($v<0$), a stationary boundary layer ($v=0$), or a bow wave ($v>0$).}
    \label{fig:TheProblem}
\end{figure}

Throughout this article, we adopt the metric signature $(-,+,+,+)$ and work in natural units, $c=\hbar=k_B=1$.

\vspace{-0.3cm}
\section{General solution via Laplace transform}
\vspace{-0.3cm}

In this section, we derive the general solution of the interface problem using the abstract geometric framework developed in \cite{GavassinoLorentzianRelxation:2026seq}. The resulting formalism applies to any linear theory with a purely relaxational excitation spectrum, including kinetic theory, radiative transfer, transient hydrodynamics, and linear viscoelasticity.

\vspace{-0.3cm}
\subsection{Model assumptions}\label{modelassumptions}
\vspace{-0.3cm}

We assume that the linearized perturbation field $\Psi:\text{``Spacetime''}\to\mathcal H$ takes values in a complex Hilbert space $\mathcal H$, endowed with the Onsager inner product $(*,*)$. Here, $\mathcal H$ denotes the space of all nonequilibrium degrees of freedom at a spacetime point, and the inner product is defined so that $\frac12(\Psi,\Psi)$ is the quadratic free-energy density of the perturbation. To simplify the derivations, we assume for now that $\dim\mathcal H<\infty$ (but arbitrarily large). Extensions to infinite-dimensional systems require additional functional-analytic assumptions, and will be discussed later.
If $\Psi$ is even under PT symmetry, the linearized equations of motion can be written in the Boltzmann-like form \cite{GavassinoUniveraalityI2023odx,GavassinoConvergence:2024xwf,GavassinoDisturbing:2026klp}
\begin{equation}\label{Boltzmann}
\partial_t\Psi=-(\s+\E\partial_x)\Psi\,,
\end{equation}
where $\s$ and $\E$ are self-adjoint operators on $\mathcal H$. The operator $\s$ is non-negative definite and models dissipation (in kinetic theory, it is the collision operator), while $\E$ has operator norm $\|\E\|\leq1$ and models propagation (in kinetic theory, it is the velocity operator) \cite{GavassinoLorentzianRelxation:2026seq}.

\vspace{-0.4cm}
\subsection{Formal solution}\label{formalsolutioon}
\vspace{-0.3cm}

Under the ansatz $\Psi=\Psi(x{-}vt)$, equation \eqref{Boltzmann} reduces to the ordinary differential equation $[\s+(\E{-}v)\partial_x]\Psi=0$. Since the solution is defined only for $x{-}vt\geq0$, the Laplace transform is the natural analogue of the Fourier decomposition for the present boundary-value problem. Typically, one requires $\Psi(x)$ to decay to zero at large $x$, corresponding to the system approaching a homogeneous equilibrium state far from the interface. In some special cases, one allows $\Psi$ to grow polynomially away from the interface (as in the Couette or Knudsen boundary layer). Either way,
\begin{equation}\label{laplaceofPsi}
\widetilde{\Psi}(s)=\int_0^\infty e^{-sx}\Psi(x)\,dx
\end{equation}
exists and is holomorphic for $\mathrm{Re}\,s>0$. The transformed equation then becomes
\begin{equation}\label{transformedequt}
[\s+(\E{-}v)s]\widetilde{\Psi}=(\E{-}v)\Psi(0)\,.
\end{equation}

We shall now use the following elementary result.

\begin{theorem}\label{theoInvert}
Let $\s$ and $\E$ be self-adjoint operators on the finite-dimensional Hilbert space $\mathcal H$, with $\s$ non-negative definite. Fix $v\in [-1,1]$, and suppose that $\ker(\s)\cap\ker(\E-v)=\{0\}$. Then $\G(s)=\s+(\E{-}v)s$ is invertible for all $s\in\mathbb C$, except possibly at finitely many points on the real axis.
\end{theorem}

\begin{proof}
Suppose, by contradiction, that there exists $s\in\mathbb C\setminus\mathbb R$ such that $\G(s)$ is not invertible. Then there exists a vector $\Psi\neq0$ such that $\G(s)\Psi=0$. Taking the inner product with $\Psi$ and separating real and imaginary parts, we obtain $(\Psi,\s\Psi)+\mathrm{Re}(s)(\Psi,(\E-v)\Psi)=0$ and $\mathrm{Im}(s)(\Psi,(\E-v)\Psi)=0$. Since $\mathrm{Im}(s)\neq0$, it follows that $(\Psi,(\E-v)\Psi)=0$, and therefore $(\Psi,\s\Psi)=0$. Since $\s$ is non-negative definite, this implies $\s\Psi=0$. Substituting this result into $\G(s)\Psi=0$, and using $s\neq0$, we also obtain $(\E-v)\Psi=0$. This contradicts the assumption that $\ker(\s)\cap\ker(\E-v)=\{0\}$. Hence, $\G(s)$ is invertible away from the real axis.

Finally, $\det\G(s)$ is a polynomial in $s$. It is not identically zero, since $\det\G(s)\neq0$ whenever $s\notin\mathbb R$. It therefore has only finitely many roots, all of which lie on the real axis. These are the points where $\G(s)$ fails to be invertible.
\end{proof}

The above theorem allows us to solve \eqref{transformedequt} away from the singular points and write
 $\widetilde{\Psi}(s)=[\s+(\E-v)s]^{-1}(\E-v)\Psi(0)$.
Note that, although the inverse matrix may be singular at isolated points on the positive real semi-axis, the vector-valued expression $[\s+(\E-v)s]^{-1}(\E-v)\Psi(0)$ must have a removable singularity there. In fact, we are assuming that $\Psi(x)$ grows at most polynomially, and therefore $\widetilde{\Psi}(s)$ is holomorphic throughout the half-plane $\mathrm{Re}\,s>0$.

Inverting the Laplace transform using Bromwich's formula therefore yields
\begin{equation}\label{laplaceformula}
\Psi(x-vt)=\int_{0^+-i\infty}^{0^++i\infty}\frac{ds}{2\pi i}\,
e^{s(x-vt)}\,[\s+(\E-v)s]^{-1}(\E-v)\Psi(0).
\end{equation}
This representation is valid for every choice of boundary data $\Psi(0)$ for which the corresponding solution grows at most polynomially, provided that $\ker(\s)\cap\ker(\E-v)=\{0\}$.

\subsection{The propagator}
\vspace{-0.3cm}

The matrix elements of $\G(s)^{-1}$ grow like powers of $s$. On the other hand, $e^{s(x-vt)}\to0$ as $\mathrm{Re}(s)\to-\infty$ for $x-vt>0$. Since the exponential factor wins, we can close the contour in equation \eqref{laplaceformula} to the left. Then, invoking \cite[\S 7.1.2, Th.~1.3]{Kato_Perturbation_Theory}, we conclude that $\G(s)^{-1}$ is holomorphic everywhere except at the singular points where $\G(s)$ is non-invertible. Hence, we can deform the contour integral into a rectangle $\Gamma$ that surrounds all the singularities of $\G(s)^{-1}$ on the non-positive real axis (see figure \ref{fig:CloseTheLoop}). This allows us to write $\Psi(x{-}vt)=\K(x{-}vt)\Psi(0)$, with
\vspace{-0.1cm}
\begin{equation}\label{Kfund}
\K(x-vt)=\oint_\Gamma\frac{ds}{2\pi i}\,
e^{s(x-vt)}\,\G(s)^{-1}(\E-v)\, ,
\end{equation}
which serves as the ``propagator'' of the theory.
This propagator gives complete information about the interface problem. In fact, we have the following result.

\begin{theorem}\label{theoK}
Let $\K(x-vt)$ be defined by \eqref{Kfund}. Suppose that $\ker(\s)\cap\ker(\E-v)=\{0\}$. Then, the following facts hold:
\begin{itemize}
\item[\textup{(i)}] $\K(x-vt)$ exists for all $x-vt$, and is entire;
\item[\textup{(ii)}] It solves equation \eqref{Boltzmann}, in the sense that $(\partial_t+\s+\E\partial_x)\K=0$;
\item[\textup{(iii)}] $\K(0)$ is a projector. It acts as the identity on the subspace of boundary data whose associated solutions grow at most polynomially, and annihilates both the boundary data associated with exponentially growing modes and the vectors in $\ker(\E-v)$, which are incompatible with the differential equation.
\end{itemize}
\end{theorem}

\begin{proof}
\textup{(i)} Since $\Gamma$ is compact and does not intersect the singular set of $\G(s)^{-1}$, the family $\G(s)^{-1}(\E{-}v)$ is holomorphic in a neighborhood of $\Gamma$. Hence, the contour integral exists in operator norm. Moreover, the exponential $e^{s(x-vt)}$ is entire in $x{-}vt$, and its Taylor series converges uniformly for $s\in\Gamma$ on every compact subset of the complex $(x{-}vt)$-plane. The sum may therefore be interchanged with the contour integral, proving that $\K(x{-}vt)$ is entire.

\textup{(ii)} Based on point \textup{(i)}, we can differentiate under the integral sign. Using $\G(s)=\s+s(\E-v)$, we therefore obtain
\begin{equation}
(\partial_t+\s+\E\partial_x)\K
=\oint_\Gamma\frac{ds}{2\pi i}\,
e^{s(x-vt)}(\E-v).
\end{equation}
Since the integrand is entire in $s$, the contour integral vanishes.

\textup{(iii)} Let us first prove that $\K(0)$ is a projector, i.e. $\K(0)^2=\K(0)$. By definition, we have
\begin{equation}\label{squadriamo}
\K(0)^2=\oint_{\Gamma_{\mathrm{out}}}\frac{ds}{2\pi i}\oint_{\Gamma_{\mathrm{in}}}\frac{dz}{2\pi i}\,\G(s)^{-1}(\E-v)\G(z)^{-1}(\E-v)\,,
\end{equation}
where we have used analyticity to deform the contour integrals so that the variable $z$ travels along a curve $\Gamma_{\mathrm{in}}$ contained inside $\Gamma_{\mathrm{out}}$, along which the variable $s$ travels. Now, we observe that $\G(s)$ obeys the resolvent-like identity
\begin{equation}
(z-s)\G(s)^{-1}(\E-v)\G(z)^{-1}(\E-v)
=\G(s)^{-1}(\E-v)-\G(z)^{-1}(\E-v).
\end{equation}
Hence, equation \eqref{squadriamo} becomes
\vspace{-0.2cm}
\begin{equation}\label{squadriamo2}
\K(0)^2=\oint_{\Gamma_{\mathrm{out}}}\frac{ds}{2\pi i}\oint_{\Gamma_{\mathrm{in}}}\frac{dz}{2\pi i}\,
\left[\frac{\G(s)^{-1}}{z-s}+\frac{\G(z)^{-1}}{s-z}\right](\E-v)\,.
\end{equation}
The $z$-integral of the first term in the square brackets vanishes by the residue theorem because $s$ lies outside the contour $\Gamma_{\mathrm{in}}$. For the second term, we exchange the order of integration and evaluate the $s$-integral using the residue theorem. Since $z$ lies inside $\Gamma_{\mathrm{out}}$, this integral returns $1$, so that $\K(0)^2=\K(0)$.

We now identify the range and kernel of $\K(0)$. Let $\Psi e^{\lambda(x-vt)}$ be a solution of \eqref{Boltzmann}. Then, $\G(\lambda)\Psi=0$, so that $\G(s)\Psi=(s{-}\lambda)(\E{-}v)\Psi$. Thus, on $\Gamma$, we have $\G(s)^{-1}(\E{-}v)\Psi=(s{-}\lambda)^{-1}\Psi$. Using the residue theorem, we obtain
\begin{equation}\label{yesorno}
\K(0)\Psi=\Psi\times
\begin{cases}
1 & \text{if }\lambda\text{ lies inside }\Gamma\, ,\\
0 & \text{if }\lambda\text{ lies outside }\Gamma\, .
\end{cases}
\end{equation}
The same residue calculation extends straightforwardly to solutions of the form $\Psi(x-vt)=\sum_n \Psi_{(n)}(x-vt)^n e^{\lambda(x-vt)}$.
Since $\Gamma$ encloses precisely the singularities with $\lambda\leq0$, $\K(0)$ acts as the identity on boundary data whose corresponding solutions grow at most polynomially as $x-vt\to+\infty$, and annihilates the exponentially growing modes associated with $\lambda>0$.

It remains to discuss the possible singularity of $\E-v$. If $\E-v$ is invertible, then the differential equation can be written as $\partial_x\Psi=-(\E-v)^{-1}\s\Psi$, and therefore admits a unique solution for every boundary value. Otherwise, let $\Phi\in\ker(\E-v)$. Taking the inner product of $[\s+(\E-v)\partial_x]\Psi=0$ with $\Phi$ yields $(\Phi,\s\Psi)=0$, which is an algebraic constraint on the boundary data. In particular, $\Phi$ itself cannot be prescribed as boundary data, since this would imply $(\Phi,\s\Phi)=0$, and hence $\s\Phi=0$, contradicting $\ker(\s)\cap\ker(\E-v)={0}$. Notably, the factor $\E-v$ in \eqref{Kfund} implies directly that $\K(0)\Phi=0$. Thus, $\K(0)$ also annihilates $\ker(\E-v)$, which corresponds to inadmissible boundary data.
\end{proof}

% Point \textup{(iii)} provides a canonical prescription for constructing admissible boundary data. Indeed, given any trial vector $\Phi\in\mathcal H$, the function $\Psi(x-vt)=\K(x-vt)\Phi$ automatically solves the equations of motion and grows at most polynomially as $x-vt\to\infty$. Thus, $\K(0)$ projects arbitrary trial boundary data onto the subspace of admissible boundary values, automatically enforcing both the differential constraints and the boundedness requirement.

\begin{figure}[h!]
    \centering
    \includegraphics[width=0.36\linewidth]{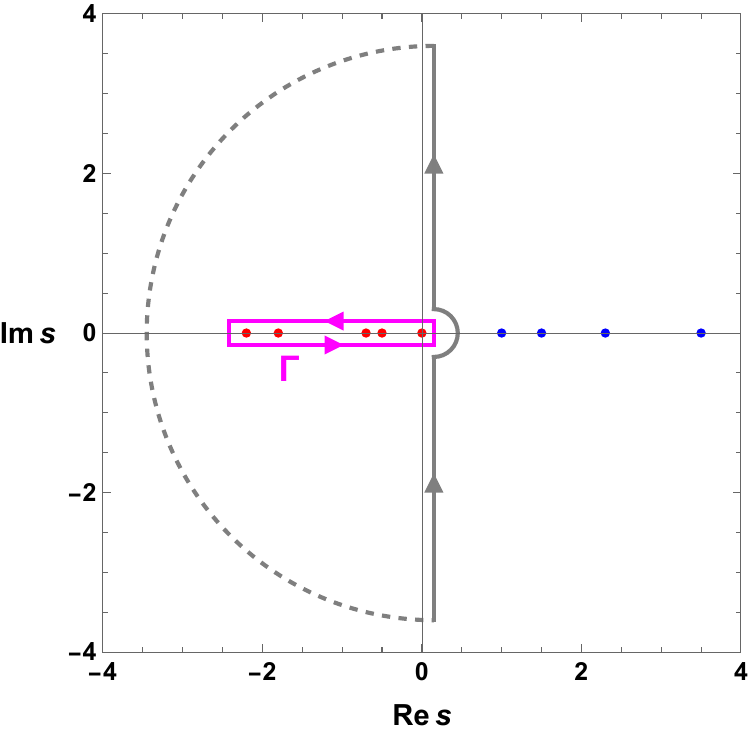}
\caption{Contour deformation in the Laplace plane. The matrix-valued function $\G(s)^{-1}$ is holomorphic everywhere except at a finite set of real singularities (red and blue). We construct a closed contour $\Gamma$ (magenta) enclosing the non-positive singularities (red), and define the propagator $\K(x-vt)$ through the contour integral \eqref{Kfund}. The resulting function $\Psi(x-vt)=\K(x-vt)\Psi(0)$ solves the interface problem with boundary data $\Psi(0)$, provided that $\Psi(0)$ is compatible with the assumption that $\Psi(x-vt)$ grows at most polynomially as $x-vt\to\infty$. Indeed, under this assumption, the vector-valued function $\G(s)^{-1}(\E-v)\Psi(0)$ has only removable singularities (blue points) on the positive real axis. The contour $\Gamma$ may therefore be deformed into the Bromwich contour (gray), recovering the inverse Laplace representation \eqref{laplaceformula}.}
    \label{fig:CloseTheLoop}
\end{figure}

\vspace{-0.3cm}
\subsection{Spectral representation of the propagator}
\vspace{-0.3cm}

Let $s_n$ denote the non-positive singularities of $\G(s)^{-1}$ (red points in Fig.~\ref{fig:CloseTheLoop}). We deform $\Gamma$ into a collection of small contours $\Gamma_n$, each enclosing a single singularity $s_n$. Equation \eqref{Kfund} then becomes $\K(x-vt)=\sum_n\K_n(x-vt)$, where
\begin{equation}\label{Kfundspectral}
\K_n(x-vt)=\oint_{\Gamma_n}\frac{ds}{2\pi i}\,
e^{s(x-vt)}\,\G(s)^{-1}(\E-v)
=
\operatorname*{Res}\!\left(
e^{s(x-vt)}\,\G(s)^{-1},\,s_n
\right)\, (\E-v),
\end{equation}
the residue of matrices being computed elementwise.

The results of Theorem~\ref{theoK} immediately extend to the operators $\K_n$. In particular, each $\K_n(x-vt)$ is entire, satisfies the equations of motion, and $\K_n(0)$ projects onto the subspace of boundary data associated with $s_n$. Indeed, elementary residue theory implies that every matrix element of $\K_n(x-vt)$ is of the form ``Polynomial$(x{-}vt)\,e^{s_n(x-vt)}$'', so that $\K_n$ isolates precisely the modes carrying the exponential factor $e^{s_n(x-vt)}$.

We can finally introduce the central observation of this paper. In general, both the number and the location of the singularities $s_n$ depend strongly on the interface velocity $v$. It is therefore useful to have a simple geometric construction that determines, for any given $v$, which modes contribute to the solution.
To this end, consider the original equation \eqref{Boltzmann}, and look for generic plane-wave solutions of the form $\Psi\propto e^{ikx-i\omega t}$, yielding the eigenvalue problem $(\s+ik\E)\Psi=i\omega\Psi$. Since the interface solution is a superposition of terms $e^{s_n(x-vt)}$, the relevant plane waves are those satisfying $e^{ikx-i\omega t}=e^{s_n(x-vt)}$, which requires $ik=s_n$ and $i\omega=vs_n$ (indeed, in the rest frame of the interface, $x'=\gamma(x-vt)$, so that $e^{s_n(x-vt)}=e^{s_n x'/\gamma}$,
which is a stationary evanescent profile \cite{GavassinoLorentzianRelxation:2026seq}). By Theorem~\ref{theoInvert}, the singularities $s_n$ are real. Hence, unlike in the usual Fourier analysis, both $ik$ and $i\omega$ are necessarily real. Conversely, whenever $ik$ is real, the operator $\s+ik\E$ is self-adjoint, so all its eigenvalues $i\omega$ are real. The spectrum may therefore be represented in the real $\{i\omega,ik\}$ plane. The singularities contributing to the interface solution are then obtained by intersecting this graph with the straight line $i\omega=vik$, retaining only the intersections with $ik\le0$. As the interface velocity varies, this line sweeps across the entire region $ik\leq-|i\omega|$, namely the left Rindler wedge of the $\{i\omega,ik\}$ plane (see also \cite{GavassinoLorentzianRelxation:2026seq}), thereby providing a complete geometric construction of all interface-localized solutions.

The interface problem is therefore reduced to a purely geometric construction on the relaxation plane.

% \vspace{-0.3cm}
% \subsection{The degenerate case}
% \vspace{-0.3cm}

% Our analysis has excluded the case $\ker(\s)\cap\ker(\E-v)\neq\{0\}$. In this situation, there exists a nonzero vector $\Phi$ such that $\s\Phi=(\E-v)\Phi=0$. Consequently, every profile of the form $\Psi=f(x-vt)\Phi$, with $f$ an arbitrary differentiable function, is an exact solution of the equations of motion. These modes are simply convected by the interface without propagating or relaxing, and are therefore completely decoupled from the spectral analysis developed above. One may then restrict the dynamics to the orthogonal complement $[\ker(\s)\cap\ker(\E-v)]^\perp$, where $\ker(\s)\cap\ker(\E-v)=\{0\}$, and carry out the previous analysis unchanged.

\section{Application to transient hydrodynamics}
\vspace{-0.3cm}

We now illustrate the formalism and the geometric construction with two transient hydrodynamic models involving only very few degrees of freedom, for which the contour integral \eqref{Kfundspectral} can be evaluated explicitly. In each example, we proceed as follows. First, we recast the equations of motion in the Onsager canonical form \eqref{Boltzmann}. Next, we represent the spectrum in the $\{i\omega,ik\}$ plane and determine, as a function of the interface velocity $v$, the modes contributing to the solution. Finally, we evaluate the propagator $\K(x-vt)$ explicitly for all possible values of $v\in [-1,1]$.

\vspace{-0.4cm}
\subsection{Cattaneo's theory of heat conduction}\label{cattaneoHeatSection}
\vspace{-0.3cm}

The dynamical variables of Cattaneo's theory \cite{cattaneo1958,Jou_Extended,GavassinoNonHydro2022} are the temperature perturbation $\delta T$ and the heat flux $\delta q$ along the $x$ direction (for clarity, we ignore the transverse components). The linearized equations of motion are
$c_v\partial_t\delta T+\partial_x\delta q{=}0$ and
$\tau\partial_t\delta q+\delta q=-\kappa\partial_x\delta T$,
where $c_v$ is the heat capacity per unit volume, $\kappa$ the thermal conductivity, and $\tau$ the relaxation time. The corresponding quadratic free-energy perturbation is $2\Delta\mathcal F=c_v(\delta T)^2/T+\tau(\delta q)^2/(T\kappa)$ \cite{GavassinoNonHydro2022}. Hence, introducing the variables
\vspace{-0.3cm}
\begin{equation}
\Psi=
\begin{bmatrix}
\sqrt{c_v/T}\,\delta T \\
\sqrt{\tau/(\kappa T)}\,\delta q \\
\end{bmatrix} \, ,
\end{equation}
the Onsager inner product becomes simply $(\Psi,\Phi)=\Psi^\dagger\Phi$, so that $\mathcal H=\mathbb C^2$.

In these variables, the equations take the canonical form \eqref{Boltzmann}:
\vspace{-0.1cm}
\begin{equation}
\begin{cases}
\partial_t\Psi_1+w\partial_x\Psi_2=0\, ,\\
\partial_t\Psi_2+\tau^{-1}\Psi_2+w\partial_x\Psi_1=0
\end{cases}
\qquad\Longrightarrow\qquad
\s+\E\partial_x=
\begin{bmatrix}
0&0\\
0&\tau^{-1}
\end{bmatrix}
+
\begin{bmatrix}
0&w\\
w&0
\end{bmatrix}
\partial_x,
\end{equation}
where $w=\sqrt{\kappa/(\tau c_v)}$ is the speed of second sound. Since $w$ is a characteristic speed, causality requires $w\le1$ \cite{Hishcock1983,rezzolla_book,GavassinoCausality2021}.

The condition $\det(\s+ik\E-i\omega)=0$ explicitly reads $i\omega\tau(i\omega\tau-1)-w^2(ik\tau)^2=0$, which defines two branches in the $\{i\omega,ik\}$ plane: the lower branch corresponds to the hydrodynamic mode, while the upper branch corresponds to the non-hydrodynamic mode. The relevant interface excitations are obtained by intersecting these branches with the half-line $i\omega=vik$, with $ik\leq0$. One intersection is always present at the origin, while the remaining intersections depend on $v$ (see Fig.~\ref{fig:CattaneoPlot}): for $0<v<w$, there is one additional hydrodynamic intersection; for $v<-w$, there is one additional non-hydrodynamic intersection; in all other cases, there are no further intersections.

Whenever an additional intersection exists, its Laplace wavenumber is always
\begin{equation}
s_c=\frac{v}{\tau \left(v^2-w^2\right)} \, .
\end{equation}

Let us compute the propagator $\K(x-vt)$ for each case.

\begin{figure}[h!]
    \centering
    \includegraphics[width=0.37\linewidth]{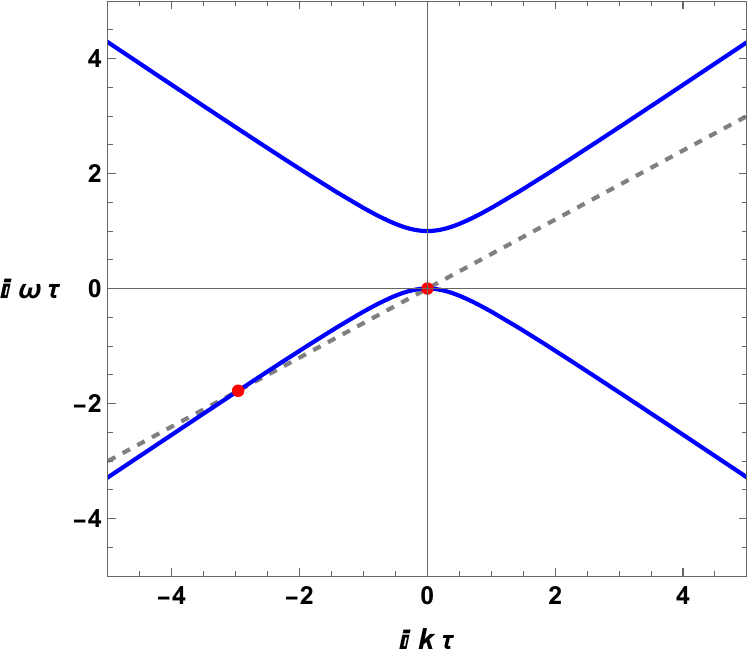}
\caption{Spectrum of Cattaneo's theory in the $\{i\omega\tau,ik\tau\}$ plane (for $w=0.75$). The blue curves are the roots of $\det(\s{+}ik\E{-}i\omega)$, with the lower branch corresponding to the hydrodynamic mode and the upper branch to the non-hydrodynamic mode. The dashed line represents the interface condition $i\omega=vik$. The red points indicate the intersections with $ik\le0$, which determine the spectral contributions to the interface solution. The value of $ik$ at each intersection is precisely the corresponding singularity $s_n$ appearing in the propagator representation \eqref{Kfundspectral}. In the example shown ($0<v<w$), there are two propagator contributions: the equilibrium mode at the origin and one exponentially decaying hydrodynamic mode.}
    \label{fig:CattaneoPlot}
\end{figure}

\subsubsection{Stationary boundary layer ($v=0$)}
\vspace{-0.3cm}

When the interface is at rest ($v=0$), the only singularity is the equilibrium mode at $s=0$. Hence, the residue formula \eqref{Kfundspectral} gives
\begin{equation}\label{cattaboundarylayer}
\K(x)= 
\begin{bmatrix}
1 & -\frac{x}{w\tau} \\
0 & 1 \\
\end{bmatrix}\, .
\end{equation}
We see that a nonzero heat flux induces a linear temperature profile. This is the familiar stationary solution of Fourier's law: when a conducting medium is placed between two plates held at different temperatures, the temperature varies linearly from the hotter plate to the colder one (within the linear regime).

\vspace{-0.3cm}
\subsubsection{Subsonic bow wave ($0<v<w$)}
\vspace{-0.3cm}

If the interface moves towards positive $x$ at a speed smaller than the speed of second sound (i.e. if $0<v<w$), then there are two singularities, at $s=0$ and at $s=s_c<0$. The residue formula \eqref{Kfundspectral} then gives
\begin{equation}\label{doubletrouble}
\K(x-vt)=
\begin{bmatrix}
1 & -\frac{w}{v}\\
0 & 0 \\
\end{bmatrix}
+ e^{\frac{v(x-vt)}{(v^2-w^2)\tau}}
\begin{bmatrix}
0 & \frac{w}{v}\\
0 & 1 \\
\end{bmatrix} \, .
\end{equation}
The first term simply shifts the asymptotic temperature and is therefore not particularly interesting. The second term is dynamical. It describes a hydrodynamic boundary layer forming ahead of the moving source. In the frame $\{t'{=}\gamma(t{-}vx),x'{=}\gamma(x{-}vt)\}$ comoving with the interface, this layer decays exponentially over the characteristic length
\begin{equation}\label{Lprimo}
L'=\dfrac{\gamma(w^2-v^2)\tau}{v}\, .
\end{equation}

\vspace{-0.3cm}
\subsubsection{Supersonic bow wave ($v\geq w$)}
\vspace{-0.3cm}

If the interface moves towards positive $x$ at a speed greater than or equal to the speed of second sound (i.e. if $v\geq w$), then the only singularity is $s=0$, and we obtain
\begin{equation}\label{supersonic}
\K(x-vt)=
\begin{bmatrix}
1 & -\frac{w}{v}\\
0 & 0 \\
\end{bmatrix}\, .
\end{equation}
This time there is no exponentially localized boundary layer. The source moves faster than the medium can transport heat, so that, if the medium is in equilibrium far from the interface, it remains in equilibrium all the way up to the instant at which the source arrives.

Unlike the previous cases, $\K(0)$ is no longer the identity. This reflects the fact that, this time, there is an exponentially growing mode, and $\K(0)$ must remove it. Indeed, Theorem~\ref{theoK} identifies $\operatorname{ran}\K(0)=\operatorname{span}{(1,0)^T}$ as the subspace of boundary data generating solutions that do not grow exponentially, while $\operatorname{ker}\K(0)=\operatorname{span}{(w,v)^T}$ is the complementary direction corresponding to the exponentially growing mode. Thus, $\K(0)$ projects arbitrary trial boundary data onto the admissible subspace by removing precisely its exponentially growing component.

\vspace{-0.3cm}
\subsubsection{Subsonic wake ($-w\leq v<0$)}
\vspace{-0.3cm}

If the interface moves towards negative $x$ at a speed smaller than the speed of second sound (i.e. if $-w<v<0$), the only singularity is again $s=0$, and the propagator is still given by \eqref{supersonic}. Thus, a slowly moving heat source leaves behind no localized wake according to Cattaneo's theory. At first sight, this may seem surprising: shouldn't heat diffuse away from the source? The resolution is that, if the temperature profile is proportional to $e^{sx}$ with $s<0$, diffusion always transports heat from hotter to colder regions, namely towards increasing $x$. Here, however, we are looking for solutions of the form $\Psi(x-vt)$ with $v<0$, which travel towards decreasing $x$. Diffusion alone cannot sustain a tail of this form.

\vspace{-0.3cm}
\subsubsection{Supersonic wake ($v< -w$)}
\vspace{-0.3cm}

If the interface moves towards negative $x$ at a speed greater than or equal to the speed of second sound (i.e. if $v<-w$), then there are two singularities, at $s=0$ and at $s=s_c<0$, and we recover \eqref{doubletrouble}. The previous argument still applies: there is no diffusive wake. However, the source now moves sufficiently fast to generate a non-hydrodynamic tail. In the frame comoving with the interface, this tail decays over the lengthscale $L'$ given in \eqref{Lprimo}. In the ultrarelativistic limit $v\to -1$, we have $L'/\tau\approx\gamma(1-w^2)$. Hence, Lorentz dilation can make this intrinsically microscopic non-hydrodynamic tail macroscopically large.

\subsection{Israel-Stewart theory at zero chemical potential}\label{ISsectionNonanaanan}

We consider a relativistic viscous fluid with no conserved charge described by Israel-Stewart theory \cite{Israel_Stewart_1979,Pu2010}. Its linearized degrees of freedom are the perturbation $\delta\varepsilon$ to the energy density, the perturbation $\delta u$ to the flow velocity along the $x$ direction (again, we ignore transverse flows), and the perturbation $\delta\Pi$ to the bulk pressure\footnote{We work with the bulk pressure to lighten the notation, but bulk viscosity and shear viscosity play the same role in the longitudinal dynamics. In particular, once recast in the canonical form \eqref{Boltzmann}, the equations of motion of a bulk-viscous fluid are identical to those of a shear-viscous fluid.}. The linearized equations of motion read $\partial_t\delta\varepsilon+(\varepsilon{+}P)\partial_x\delta u=0$, $(\varepsilon{+}P)\partial_t\delta u+c_s^2\partial_x\delta\varepsilon+\partial_x\delta\Pi=0$, and $\tau\partial_t\delta\Pi+\delta\Pi=-\zeta\partial_x\delta u$, where $P$ is the pressure, $c_s^2$ the speed of sound, $\zeta$ the bulk viscosity coefficient, and $\tau$ the relaxation time. The corresponding quadratic free-energy perturbation is $2\Delta\mathcal F=c_s^2(\delta\varepsilon)^2/(\varepsilon{+}P)+(\varepsilon{+}P)(\delta u)^2+\tau(\delta\Pi)^2/\zeta$ \cite{GavassinoUniveraalityI2023odx,MullinsInfo2023tjg}. Hence, introducing
\begin{equation}
\Psi=
\begin{bmatrix}
\sqrt{c_s^2/(\varepsilon{+}P)}\,\delta\varepsilon\\
\sqrt{\varepsilon{+}P}\,\delta u\\
\sqrt{\tau/\zeta}\,\delta\Pi
\end{bmatrix},
\end{equation}
the Onsager inner product becomes simply $(\Psi,\Phi)=\Psi^\dagger\Phi$, so that $\mathcal H=\mathbb C^3$, and the equations of motion reduce to
\begin{equation}
\begin{cases}
\partial_t\Psi_1+c_s\partial_x\Psi_2=0\, ,\\
\partial_t\Psi_2+c_s\partial_x\Psi_1+c_a\partial_x\Psi_3=0\, ,\\
\partial_t\Psi_3+\tau^{-1}\Psi_3+c_a\partial_x\Psi_2=0
\end{cases}
\qquad\Longrightarrow\qquad
\s+\E\partial_x=
\begin{bmatrix}
0&0&0\\
0&0&0\\
0&0&\tau^{-1}
\end{bmatrix}
+
\begin{bmatrix}
0&c_s&0\\
c_s&0&c_a\\
0&c_a&0
\end{bmatrix}
\partial_x,
\end{equation}
where $c_a=\sqrt{\zeta/[\tau(\varepsilon+P)]}$. Causality requires the characteristic speed $w=\sqrt{c_s^2+c_a^2}$ to be subluminal \cite{Pu2010}.

The dispersion relation $\det[\s+ik\E-i\omega]=0$ takes the form $(i\omega\tau)^3-(i\omega\tau)^2-w^2(i\omega\tau)(ik\tau)^2+c_s^2(ik\tau)^2=0$, and is shown in Fig.~\ref{fig:ISPlot}. As in Cattaneo's theory, one intersection is always present at the origin, while the remaining intersections depend on the interface velocity. A simple geometric inspection reveals that, for $-c_s<v<0$, there is one additional intersection with the left-moving sound branch (i.e. the mode with $i\omega=-c_sik+\mathcal O(k^2)$), for $c_s<v<w$, there is one additional intersection with the right-moving sound branch (i.e. the mode with $i\omega=c_sik+\mathcal O(k^2)$), and for $v<-w$, there is one additional intersection with the non-hydrodynamic branch. In all other cases, there are no further intersections.

Notably, whenever an additional intersection exists, its Laplace wavenumber is always
\begin{equation}\label{scIS}
s_c=\frac{c_s^2-v^2}{\tau v(w^2-v^2)}\,.
\end{equation}

As in the Cattaneo case, we now examine each velocity regime separately.

\begin{figure}[h!]
    \centering
    \includegraphics[width=0.38\linewidth]{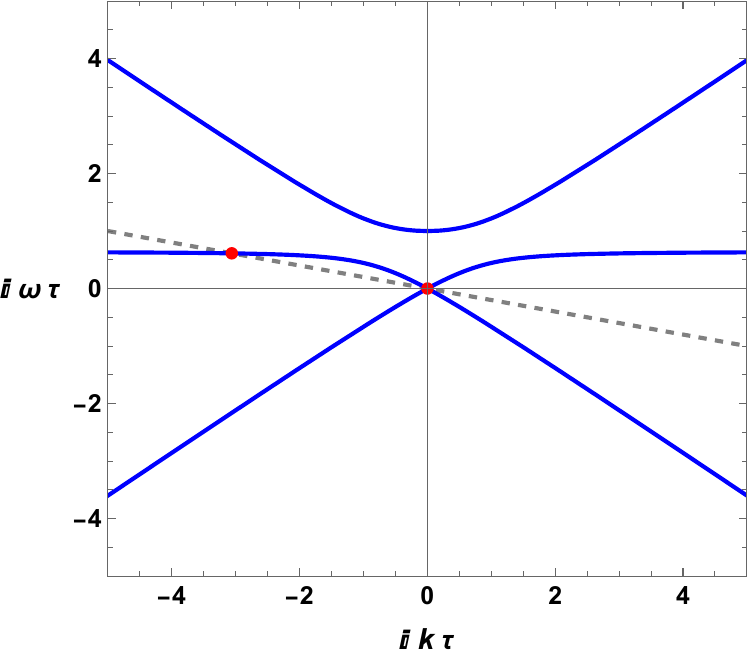}
\caption{Spectrum of Israel-Stewart theory with bulk viscosity in the $\{i\omega\tau,ik\tau\}$ plane (for $c_s=0.6$ and $w=0.75$). The blue curves are the solutions of the dispersion relation $\det(\s{+}ik\E{-}i\omega)=0$, with the lower branches corresponding to the left- and right-moving sound modes, and the upper branch to the non-hydrodynamic mode. The dashed line represents the interface condition $i\omega=vik$. The red points indicate the intersections with $ik\le0$, which determine the spectral contributions to the interface solution. In the example shown ($-c_s<v<0$), there are two contributions: the equilibrium mode at the origin and one exponentially decaying left-moving sound mode.}
    \label{fig:ISPlot}
\end{figure}

\subsubsection{Stationary boundary layer ($v=0$)}

When the interface is at rest ($v=0$), the only singularity is at $s=0$, and we obtain
\begin{equation}
\K(x)=
\begin{bmatrix}
1 & 0 & \frac{c_a}{c_s} \\
0 & 1 & 0 \\
0 & 0 & 0
\end{bmatrix}\, .
\end{equation}
This shows that a fluid at zero chemical potential cannot sustain a nontrivial longitudinal boundary layer in the hydrodynamic regime: the general solution $\K(x)\Phi$ is independent of $x$. By contrast, fluids at finite chemical potential can exhibit stationary thermal gradients, with a temperature that varies linearly in $x$, as illustrated by the Cattaneo solution \eqref{cattaboundarylayer}.

Since $\operatorname{ran}\K(0)=\operatorname{span}\{(1,0,0)^T,(0,1,0)^T\}$, every admissible boundary value satisfies $\delta\Pi=0$. On the other hand, the vector $(-c_a,0,c_s)^T\in\ker\K(0)$ does not correspond to an exponentially growing mode; rather, it spans $\ker\E$. In agreement with Theorem~\ref{theoK}, its removal reflects an algebraic constraint imposed by the stationary equations. Indeed, taking the inner product of $[\s+\E\partial_x]\Psi=0$ with $(-c_a,0,c_s)^T$ gives $\Psi_3=0$. Thus, boundary data with $\delta\Pi\neq0$ cannot be extended to a stationary solution.

\subsubsection{Subsonic bow wave ($0<v<c_s$)}

When the interface moves subsonically to the right ($0<v<c_s$), $s=0$ is still the only singularity, and we find
\begin{equation}\label{subsonicbow}
\K(x-vt)=
\begin{bmatrix}
1 & 0 & \frac{c_s c_a}{c_s^2-v^2} \\
0 & 1 & \frac{v c_a}{c_s^2-v^2} \\
0 & 0 & 0
\end{bmatrix}\, .
\end{equation}
Again, the fluid cannot sustain any stationary localized gradient, nor any viscous correction: the general solution $\K(x{-}vt)\Phi$ is independent of $x{-}vt$, and has $\delta\Pi=0$. Physically, the interface moves more slowly than the speed of sound, so any sound wave generated at the interface propagates away faster than the interface itself. As a result, no localized bow wave can form.

Unlike the stationary case, however, $\E-v$ is now invertible, so the differential equation admits a unique solution for every boundary datum. This time, $\K(0)$ projects out a genuine exponentially growing mode. In particular, the vector $(-c_s c_a,-vc_a,c_s^2{-}v^2)^T$ spans $\ker\K(0)$ and generates a solution that grows exponentially away from the interface. Physically, this mode represents a right-moving sound wave. Left to itself, it would propagate at the sound speed $c_s>v$. The exponentially growing profile generates a diffusive flux of momentum from right to left, which continuously replenishes the wave and reduces the propagation speed of its envelope down to $v$.

\subsubsection{Sonic bow wave ($v=c_s$)}

When the interface moves to the right exactly at the speed of sound ($v=c_s$), $s=0$ is again the only singularity. However, this time we have
\begin{equation}\label{theone}
\K(x-vt)=
\begin{bmatrix}
1 & 0 & -\frac{x-vt}{\tau c_a} \\
0 & 1 & -\frac{x-vt}{\tau c_a} \\
0 & 0 & 1
\end{bmatrix}\, .
\end{equation}
Unlike the previous case, $\K(0)$ is now the identity, so every boundary datum is admissible. Moreover, a nonzero viscous pressure generates linear gradients in both the energy density and the flow velocity.
This reflects the fact that $v=c_s$ is a critical velocity. As $v\to c_s^-$, the growth rate of the excluded sound mode tends to zero. At $v=c_s$, it merges with the equilibrium mode at $s=0$, the exponential modulation is replaced by a linear profile, and $\K(0)$ becomes the identity.

\subsubsection{Supersonic subcharacteristic bow wave ($c_s<v<w$)}

Once the interface velocity exceeds the speed of sound, but remains below the characteristic speed $w$, the additional singularity \eqref{scIS} crosses the origin and becomes negative. As a result, the propagator acquires two contributions:
\begin{equation}\label{propagoSpersonic}
\K(x-vt)=
\begin{bmatrix}
1 & 0 & \frac{c_s c_a}{c_s^2-v^2} \\
0 & 1 & \frac{v c_a}{c_s^2-v^2} \\
0 & 0 & 0
\end{bmatrix}
+e^{-\frac{(v^2-c_s^2)(x-vt)}{\tau v(w^2-v^2)}}
\begin{bmatrix}
0 & 0 & -\frac{c_s c_a}{c_s^2-v^2} \\
0 & 0 & -\frac{v c_a}{c_s^2-v^2} \\
0 & 0 & 1
\end{bmatrix}\, .
\end{equation}
The first term is the equilibrium contribution, while the second is a right-moving sound mode with an exponentially decaying profile. In the homogeneous medium, this wave would propagate at the sound speed $c_s<v$, and would therefore lag behind the interface. The exponential modulation, however, creates a diffusive flux of momentum from left to right. This continuous transfer of momentum accelerates the envelope until it propagates at the interface velocity $v$, making the profile stationary in the comoving frame.

The solution \eqref{propagoSpersonic} admits a natural interpretation as the asymptotic viscous tail of a shock wave. Indeed, the condition $v>c_s$ is precisely the Lax shock condition for a right-moving shock propagating into a homogeneous fluid at rest. The interface at $x=vt$ should therefore be viewed as the point where the disturbance has become sufficiently weak for the linear approximation to apply, while the actual shock front lies further behind, at some position $x=-d+vt$, with $d\gg |s_c|^{-1}$. The exponentially decaying contribution is then the final part of the viscous regularization of the shock front.

In the frame $\{t'=\gamma(t-vx),x'=\gamma(x-vt)\}$ comoving with the interface, the tail decays over the characteristic length
\begin{equation}
L'
=
\frac{\gamma v(w^2-v^2)\tau}{v^2-c_s^2}\,.
\end{equation}
As $v\to c_s^+$, corresponding to the weak-shock limit, $L'$ diverges. The gradients therefore become arbitrarily small, and the tail lies entirely within the hydrodynamic regime. Expanding near $v=c_s$, we find
\begin{equation}
L'
=
\frac{\gamma\zeta}{2(\varepsilon+P)(v-c_s)}
+\mathcal O(1),
\end{equation}
which is independent of the relaxation time $\tau$. Thus, in the weak-shock limit, the asymptotic tail is universal and depends only on the Navier-Stokes transport coefficient $\zeta$. By contrast, as $v\to w^-$, the decay length shrinks to zero, and the tail probes microscopic scales. In this regime, its detailed structure depends on the particular transient theory used to regularize hydrodynamics, with Israel-Stewart theory providing only one possible realization.

\subsubsection{Supercharacteristic bow wave ($v\geq w$)}

When the interface velocity reaches the characteristic speed $w$, the singularity $s_c$ escapes to $-\infty$. For $v>w$, it reappears on the positive real axis, where it is excluded from the contour integral. Hence, the propagator reduces again to \eqref{subsonicbow}: the fluid can no longer sustain stationary localized gradients or viscous corrections.

We have recovered a well-known result: Israel-Stewart theory cannot regularize supercharacteristic shocks with a smooth viscous profile \cite{OlsonShock1990,BemficaShock2025}. Geometrically, the reason is transparent. Once $v$ exceeds the characteristic speed, the line $i\omega=vik$ no longer intersects the spectrum in the left half-plane, so no exponentially decaying mode is available to form the asymptotic shock tail. More generally, this is a manifestation of the hyperbolic nature of Israel-Stewart theory: supercharacteristic shocks are necessarily non-smooth.

This conclusion extends well beyond Israel-Stewart theory. Indeed, it was shown in \cite{GavassinoLorentzianRelxation:2026seq} that, for every system of the form \eqref{Boltzmann}, the graph of the spectrum in the $\{i\omega,ik\}$ plane has slope everywhere bounded by $w=||\E||$, namely the fastest characteristic speed. Consequently, if $v>w$, the line $i\omega=vik$ cannot intersect the spectrum for $ik<0$, and therefore no exponentially decaying tail can exist. It follows that only theories with $w=1$, such as kinetic theory, are capable of regularizing arbitrarily strong relativistic shocks. 

\subsubsection{Subsonic wake ($-c_s<v<0$)}

When the interface moves to the left at a speed smaller than the speed of sound ($-c_s<v<0$), the additional singularity again lies on the negative real axis, so the propagator is still given by \eqref{propagoSpersonic}. The exponentially decaying contribution now corresponds to the left-moving sound wave.

This solution naturally describes the asymptotic tail behind a shock wave. Indeed, in the frame where the fluid behind the shock is at rest, the shock velocity is typically subsonic. The interface $x=vt$ should therefore be viewed as the point where the disturbance has become sufficiently weak for the linear approximation to apply, while the shock front itself lies further to the left. The exponentially decaying contribution is then the hydrodynamic wake left behind by the shock.
Physically, the left-moving sound wave would normally propagate at speed $-c_s<v$, and would therefore outrun the interface. By acquiring an exponentially decaying profile, however, momentum diffusion continuously transfers momentum from the compressed region ahead towards the wake, reducing the propagation speed of the envelope until it matches the interface velocity $v$.

In the frame $\{t'=\gamma(t{-}vx),x'=\gamma(x{-}vt)\}$ comoving with the interface, the wake decays over the lengthscale
\begin{equation}\label{thelui}
L'
=
\frac{\gamma\tau |v|(w^2-v^2)}{c_s^2-v^2}\,.
\end{equation}
As $v\to -c_s^+$, corresponding to the weak-shock limit, $L'\to\infty$, so the wake becomes arbitrarily broad and is entirely described by hydrodynamics. Expanding about $v=-c_s$, one again finds that the leading contribution to $L'$ depends only on the viscosity coefficient $\zeta$, and is independent of the relaxation time $\tau$. By contrast, as $v\to0^-$, the singularity $s_c$ escapes to $-\infty$, and $L'\to0$. The wake then probes microscopic scales, and its detailed structure depends on the particular transient theory used to regularize hydrodynamics.

\subsubsection{Sonic wake ($v=-c_s$)}

As $v\to-c_s^+$, the singularity $s_c$ approaches the origin and merges with the equilibrium singularity at $s=0$. At the sonic point, the exponentially decaying wake is therefore replaced by a linearly varying profile, and the propagator becomes
\begin{equation}\label{theone}
\K(x-vt)=
\begin{bmatrix}
1 & 0 & \frac{x-vt}{\tau c_a} \\
0 & 1 & -\frac{x-vt}{\tau c_a} \\
0 & 0 & 1
\end{bmatrix}\, .
\end{equation}
In particular, $\K(0)=1$, so every boundary datum is admissible. A nonzero viscous perturbation generates opposite linear gradients in the first two hydrodynamic variables. Thus, precisely at the sonic velocity, the exponential wake delocalizes into a polynomially growing solution, reflecting the coalescence of the left-moving sound mode with the equilibrium mode.

\subsubsection{Supersonic subcharacteristic wake ($-w\leq v<-c_s$)}

When the interface velocity decreases below $-c_s$, the additional singularity $s_c$ crosses the origin and becomes positive. It is therefore excluded from the contour integral, and the propagator reduces again to \eqref{subsonicbow}. Thus, although the interface remains subcharacteristic, no additional exponentially decaying mode contributes: the fluid cannot sustain a stationary localized wake or viscous correction.

The excluded mode is a left-moving sound wave with an exponentially growing profile. In the homogeneous medium, it would propagate at the sound speed $-c_s$, which is slower than the interface velocity $v<-c_s$. The exponential growth from left to right causes viscosity to continuously transfer momentum to the left, accelerating the envelope until it propagates at the interface velocity $v$. Since this requires an exponentially growing profile, however, the mode is excluded from the contour integral, and only the homogeneous equilibrium solution remains.

\subsubsection{Supercharacteristic wake ($v< -w$)}

Once the interface velocity exceeds the characteristic speed, the additional singularity re-enters the negative real axis, so the propagator is again given by \eqref{propagoSpersonic}. This time, however, the singularity $s_c$ belongs to the non-hydrodynamic branch, so the exponentially decaying contribution represents a non-hydrodynamic wake trailing behind the interface.

The exponential decay length in the frame $\{t'=\gamma(t-vx),x'=\gamma(x-vt)\}$ comoving with the interface is still \eqref{thelui}. In contrast to the weak-shock regime, this large length scale is not produced by a sound-mode singularity approaching the origin. Instead, the non-hydrodynamic mode retains a microscopic relaxation timescale of order $\tau$, while the corresponding spatial scale in the comoving frame is stretched by the Lorentz factor. Hence, as $v\to-1$, $\gamma\to\infty$ and consequently $L'\to\infty$. Thus, even though the wake remains intrinsically non-hydrodynamic, Lorentz dilation can make it macroscopically large in the frame comoving with the interface.

\vspace{-0.3cm}
\section{Application to systems with many degrees of freedom}
\vspace{-0.3cm}

As the dimension of $\mathcal{H}$ increases, evaluating the propagator $\K(x-vt)$ exactly becomes increasingly difficult. Indeed, the singularities of $\G(s)^{-1}$ become the roots of a high-degree characteristic polynomial, which generically cannot be expressed in closed form. Nevertheless, the geometric construction developed in this work remains fully applicable. Combined with the bounds of \cite{GavassinoLorentzianRelxation:2026seq}, it provides rigorous information about the structure of the interface solutions and about the location of the singularities, which determine the exponential factors appearing in the asymptotic tails.

Moreover, if one is interested only in these exponential factors, there is no need to restrict to finite-dimensional systems. In Appendix~\ref{AAA}, we show that, under suitable convergence assumptions (including a spectral gap for the non-hydrodynamic modes), one can still define a propagator $\K(x-vt)$ that characterizes the interface solutions. Its matrix elements admit the spectral representation
\begin{equation}\label{spectralrepresentation}
(\Phi_1,\K(x-vt)\Phi_2)
=
\int_{-\infty}^{0^+} ds\,\rho(s)e^{s(x-vt)},
\end{equation}
where the spectral density $\rho(s)$ is supported on those values of $s$ for which the point $(i\omega,ik)=(vs,s)$ belongs to the spectrum of the theory\footnote{Note that $\rho(s)$ is in general a distribution, so the representation \eqref{spectralrepresentation} can easily incorporate polynomial prefactors multiplying the exponentials. For instance, a term in $\rho(s)$ proportional to $\delta^{(n)}(s-s_0)$ gives rise to a term $\propto (x-vt)^n e^{s_0(x-vt)}$ in the matrix element.}.

In the remainder of this section, we illustrate this construction with two examples: a finite-dimensional theory with many degrees of freedom, and an infinite-dimensional theory.

\vspace{-0.3cm}
\subsection{Cattaneo theory with multiple heat fluxes}
\vspace{-0.3cm}

As a simple application of the formalism, let us consider a generalization of Cattaneo's theory of heat conduction (see subsection \ref{cattaneoHeatSection}), where the total heat flux is the sum of $N$ contributions, each with its own conductivity and relaxation time. The equations of motion then read
$c_v \partial_t \delta T+\sum_n \partial_x \delta q_n=0$
and
$\tau_n \partial_t \delta q_n+\delta q_n=-\kappa_n \partial_x \delta T$.
The corresponding quadratic free-energy density perturbation is
$2\Delta\mathcal F=c_v(\delta T)^2/T+\sum_n\tau_n(\delta q_n)^2/(T\kappa_n)$
\cite{GavassinoNonHydro2022}. Hence, setting
$\Psi=(\sqrt{c_v/T}\,\delta T,\sqrt{\tau_1/(\kappa_1T)}\,\delta q_1,\sqrt{\tau_2/(\kappa_2T)}\,\delta q_2,\ldots)$,
we obtain $(\Psi,\Phi)=\Psi^\dagger\Phi$ and $\mathcal H=\mathbb C^{1+N}$.

In these variables, we find that $\s$ and $\E$ have the block form
\begin{equation}
\s=
\begin{bmatrix}
0 & 0 \\
0 & \mathbb{T}^{-1}
\end{bmatrix},
\qquad\qquad
\E=
\begin{bmatrix}
0 & \mathbb{V}^T \\
\mathbb{V} & 0
\end{bmatrix},
\end{equation}
where $\mathbb{T}=\mathrm{diag}(\tau_1,\tau_2,\ldots)$ and $\mathbb{V}=(c_1,c_2,\ldots)^T$, with $c_n=\sqrt{\kappa_n/(\tau_n c_v)}$. Causality requires that the characteristic speed $w=\sqrt{\mathbb{V}^T\mathbb{V}}$ does not exceed $1$.

Let us now study the structure of the spectrum of this theory in the $\{i\omega,ik\}$ plane. For clarity, we assume that all $\tau_n$ are distinct and all $c_n$ are non-zero. The condition $\det(\s+ik\E-i\omega)=0$ can be written explicitly using block-matrix theory, giving
\begin{equation}\label{thePoly}
\left[i\omega+\sum_n \dfrac{(c_n ik)^2}{\tau_n^{-1}-i\omega}\right]
\prod_m(\tau_m^{-1}-i\omega)=0.
\end{equation}
This is a polynomial of degree $1+N$ in $i\omega$, and therefore defines $1+N$ dispersion relations $i\omega_n(ik)$. For $ik=0$, the roots are simply $0$ and $\tau_n^{-1}$, corresponding to one hydrodynamic mode and $N$ non-degenerate non-hydrodynamic modes. As $ik$ is turned on, these modes define $1+N$ analytic dispersion curves in a neighborhood of the real $ik$-axis.

Now suppose that the relaxation times are ordered so that $\tau_1^{-1}<\tau_2^{-1}<\cdots<\tau_N^{-1}$, and consider the square bracket in \eqref{thePoly} for $ik\neq0$. Each term $(c_nik)^2/(\tau_n^{-1}{-}i\omega)$ diverges to $+\infty$ as $i\omega$ approaches $\tau_n^{-1}$ from below, and to $-\infty$ as it approaches from above. Thus, between any two consecutive relaxation rates $\tau_{n}^{-1}$ and $\tau_{n+1}^{-1}$, the square bracket changes sign and therefore possesses a zero. We conclude that there are exactly $N{-}1$ roots $i\omega_n$, each lying between two consecutive relaxation rates. The remaining two roots become unbounded as $|ik|{\to}\infty$. Indeed, isolating $(ik)^2$, we obtain
\begin{equation}
(ik)^2
=
\dfrac{i\omega}{\sum_n c_n^2(i\omega-\tau_n^{-1})^{-1}}
\xrightarrow{|i\omega|\to\infty}
\dfrac{(i\omega)^2}{w^2},
\end{equation}
so that the two remaining branches satisfy $i\omega\approx\pm w|ik|$ asymptotically. Since the hydrodynamic branch is connected to the root $i\omega=0$, it follows that
$i\omega\approx-w|ik|$,
while the remaining non-hydrodynamic branch satisfies
$i\omega\approx w|ik|$.
Thus, independently of the number of relaxation channels, only one non-hydrodynamic branch reaches the characteristic speed $w$, while all the others remain confined between adjacent relaxation rates.

An explicit example is shown in figure \ref{fig:CattaneoMultple}.

\begin{figure}[h!]
    \centering
    \includegraphics[width=0.38\linewidth]{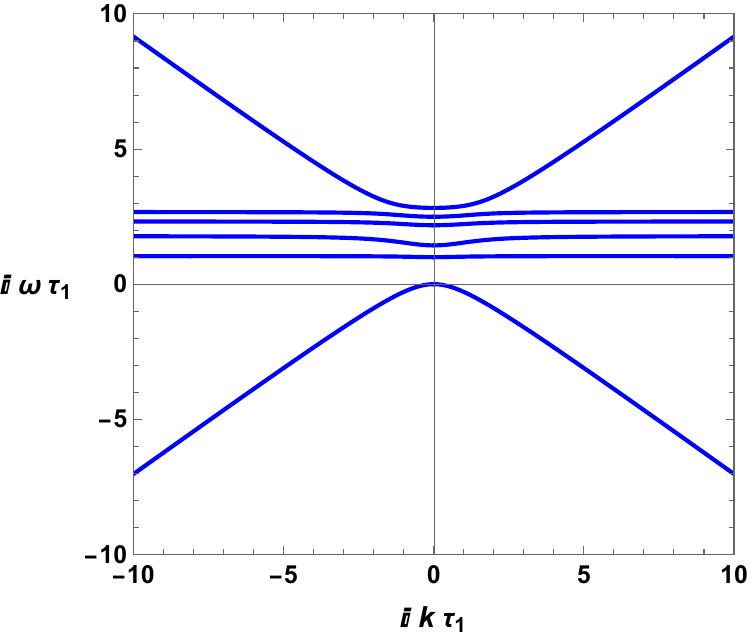}
\caption{Dispersion relations of the generalized Cattaneo theory with five relaxation channels. The relaxation times and conductivities were chosen randomly, subject to the constraint $w=0.8$. Frequencies and wave numbers are expressed in units of the slowest relaxation time $\tau_1$. The figure illustrates the generic structure predicted by equation \eqref{thePoly}: one hydrodynamic branch and one non-hydrodynamic branch asymptotically approach the characteristic cone $i\omega=\pm wik$, whereas the remaining non-hydrodynamic branches remain trapped between consecutive relaxation rates.}
    \label{fig:CattaneoMultple}
\end{figure}

\newpage
We can use the above spectral characterization to describe the solutions of the interface problem without explicitly determining the intersections with the line $i\omega=vik$, as we now illustrate.

For stationary boundary layers ($v=0$), the only intersection occurs at the origin. In this particular case, the propagator can still be computed explicitly, yielding
\begin{equation}
\K(x)=
\begin{bmatrix}
1 & & -\dfrac{\mathbb{V}^T x}{\mathbb{V}^T \mathbb{T}\mathbb{V}} \\
\\
0 & & \dfrac{\mathbb{T}\mathbb{V}\mathbb{V}^T}{\mathbb{V}^T \mathbb{T}\mathbb{V}}
\end{bmatrix} \, ,
\end{equation}
which is similar to the propagator of the original Cattaneo theory: the temperature develops a linear gradient whenever the total heat flux is non-vanishing. The structure of the bow waves ($v>0$) is also similar to the Cattaneo case: There is a single exponentially decaying tail for subcharacteristic interfaces ($v<w$), while no tail exists for supercharacteristic motions ($v>w$).

The novelties with respect to the original Cattaneo theory arise for wakes ($v<0$). In this case, there are at least $N-1$ exponentially decaying tails for every interface velocity. Indeed, the line $i\omega=vik$ passes through each band $\tau_n^{-1}<i\omega<\tau_{n+1}^{-1}$ for $(v\tau_n)^{-1}>ik>(v\tau_{n+1})^{-1}$, and must therefore intersect the corresponding dispersion relation at least once. Hence, the matrix elements $(\Phi_1,\K(x-vt)\Phi_2)$ contain a linear combination of exponential factors $e^{s_n(x-vt)}$, each of which decays in the frame comoving with the interface over a characteristic length
\begin{equation}
L'_n\in(\gamma|v|\tau_{n+1},\,\gamma|v|\tau_n)\, .
\end{equation}
This estimate is especially useful when the theory contains many degrees of freedom, so that the relaxation times $\tau_n$ are densely spaced. In that case, the interval becomes very narrow, and one simply has $L'_n\approx\gamma|v|\tau_n$.

\subsection{Kinetic theory in the relaxation-time approximation}

As an example of a system with $\dim\mathcal{H}=\infty$, we consider an ultrarelativistic bosonic gas at zero chemical potential, modeled by kinetic theory in the Relaxation-Time Approximation (RTA) \cite{AndersonWitting1974,cercignani_book}. For simplicity, we restrict attention to momentum distribution functions that are isotropic in the transverse plane, so that all flows are longitudinal. In Appendix~\ref{BBB}, we show that, under a suitable assumption on the energy dependence of the perturbed distribution function, one may take $\mathcal{H}=L^2([-1,1])$, endowed with the Onsager inner product
$(\Psi,\Phi)=\int_{-1}^1\Psi(\Omega)^*\Phi(\Omega)\,d\Omega$,
where $\Omega$ denotes the $x$-component of the particle velocity, and $\Psi$ is the relevant velocity distribution. We also show that the linearized Boltzmann equation takes the form
\begin{equation}
(\partial_t+\Omega\partial_x)\Psi
=
\tau^{-1}\left[
\dfrac{(1,\Psi)}{(1,1)}
+\Omega\dfrac{(\Omega,\Psi)}{(\Omega,\Omega)}
-\Psi
\right] .
\end{equation}

The spectrum is determined by the eigenvalue problem $(\s+ik\E)\Psi=i\omega\Psi$. It is immediate to see that the operator $\s+ik\E$ differs from the multiplication operator $\tau^{-1}+ik\Omega$ by a rank-two projector. Hence, its spectrum consists of the continuous part $i\omega\in\tau^{-1}+ik[-1,1]$, together with at most two discrete eigenvalues, corresponding to the left- and right-moving sound modes. The discrete modes satisfy
\vspace{-0.1cm}
\begin{equation}
\Psi
=\dfrac{1}{1-i\omega\tau+ik\tau\Omega} \left[
\dfrac{(1,\Psi)}{(1,1)}
+\Omega\dfrac{(\Omega,\Psi)}{(\Omega,\Omega)}\right] \, ,
\end{equation}
and exist provided the resulting function $\Psi$ belongs to $L^2([-1,1])$. This is the case if and only if $i\omega$ lies outside the continuous spectrum. Taking the inner product of this equation with $1$ and $\Omega$ gives
\vspace{-0.1cm}
\begin{equation}
\begin{bmatrix}
\text{arctanh}\left( \dfrac{ik\tau}{1{-}i\omega\tau}\right)-ik\tau & & & 3-3\dfrac{1{-}i\omega\tau}{ik\tau}\text{arctanh}\left(\dfrac{ik\tau}{1{-}i\omega\tau} \right)  \\
\dfrac{1{-}i\omega\tau}{ik\tau}\text{arctanh} \left(\dfrac{ik\tau}{1{-}i\omega\tau} \right)-1 & & & ik\tau+3\dfrac{1{-}i\omega\tau}{ik\tau}-3\dfrac{(1{-}i\omega\tau)^2}{(ik\tau)^2} \text{arctanh} \left(\dfrac{ik\tau}{1{-}i\omega\tau} \right)  \\
\end{bmatrix}
\begin{bmatrix}
(1,\Psi) \\
\\
(\Omega,\Psi) \\
\end{bmatrix}=0 \, .
\end{equation}
Setting the determinant of the $2\times 2$ matrix to $0$, we obtain an implicit function $F(i\omega,ik)=0$. Remarkably, the solutions can be written explicitly as parametric curves. In particular, by setting $r=ik\tau/(1-i\omega\tau)$, we obtain
\vspace{-0.1cm}
\begin{equation}\label{parametric}
\begin{split}
ik(r)\tau={}& \frac{1}{2} \left[\frac{3 \text{arctanh}(r)}{r^2}\pm\sqrt{\left(\frac{3 \text{arctanh}(r)}{r^2}-\frac{3}{r}+\text{arctanh}(r)\right)^2-4 \left(\frac{3 \text{arctanh}(r)}{r}-3\right)}-\frac{3}{r}+\text{arctanh}(r)\right] \, , \\
i\omega(r)\tau ={}& 1-ik(r)\tau/r \, . \\
\end{split}
\end{equation}
As a consistency check, we expand the parametric solution for small $r$ and match it to the hydrodynamic sound-wave dispersion relation
$i\omega=\mp c_sik-\frac{2\eta}{3(\varepsilon+P)}(ik)^2+\cdots$.
This yields the speed of sound of an ultrarelativistic gas, $c_s=1/\sqrt3$, together with the shear viscosity $\eta=\frac15(\varepsilon+P)\tau$, in agreement with the standard RTA result.

The spectrum of the theory is plotted in Figure~\ref{fig:RTASOUND}. Taking its intersections with the line $i\omega=vik$, we find that, for all subluminal interface velocities, the matrix elements of the propagator always contain a continuous contribution:
\vspace{-0.1cm}
\begin{equation}\label{spectralrRTA}
(\Phi_1,\K(x-vt)\Phi_2)
=\text{``Discrete part''}+
\int_{-\infty}^{-\tau^{-1}(1-v)^{-1}} ds\,\rho(s)e^{s(x-vt)},
\end{equation}
so that RTA kinetic theory can smooth arbitrarily strong shocks. In the frame comoving with the interface, the continuous contribution decays over a characteristic length of at most $L'{=}\gamma\tau(1{-}v)$. As $v{\to}1$ (near-luminal bow wave), $L'{\to}0$, so the tail in front of the interface produced by the continuous spectrum becomes infinitely thin. Conversely, as $v{\to}-1$ (near-luminal wake), $L'{\to}\infty$, so the tail left behind the interface becomes macroscopically large.

\begin{figure}[h!]
    \centering
    \includegraphics[width=0.38\linewidth]{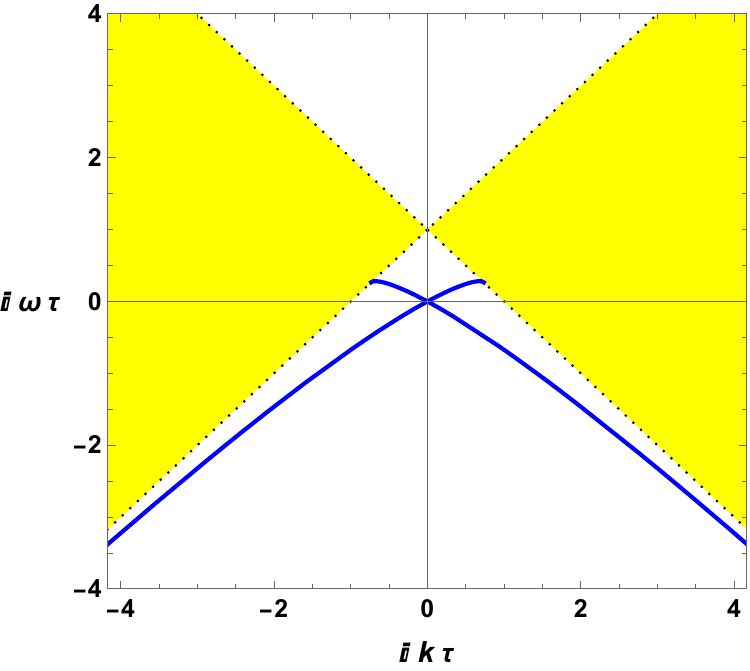}
\caption{Spectrum of RTA kinetic theory in the $\{i\omega\tau,ik\tau\}$ plane. The yellow region represents the continuous spectrum $i\omega\tau\in1+ik\tau[-1,1]$. The blue curves are the two sound branches obtained from the parametric solution \eqref{parametric}, which is defined only for $r\in[-1,1]$. As shown in the figure, the branches meet the continuous spectrum at $(i\omega\tau,ik\tau)=(1/4,\pm3/4)$ and bend so as to become tangent to its boundary at these points. Their unshown tails extend to $i\omega\to-\infty$ and $ik\tau\to\pm\infty$, asymptotically approaching $i\omega=-|ik|$.}
    \label{fig:RTASOUND}
\end{figure}

Finally, let us comment briefly on the discrete contributions to the propagator. The intersection at $(i\omega,ik)=(0,0)$ is always present, and describes global shifts of the equilibrium state. For $v>0$ (bow waves), the situation is analogous to the Israel--Stewart case discussed in section \ref{ISsectionNonanaanan} with $w=1$: there are no exponentially decaying discrete modes up to $v=1/\sqrt{3}=c_s$. Above this threshold (supersonic bow waves), a single exponentially decaying discrete mode appears, corresponding to the intersection with the right-moving sound branch. This mode disappears only at $v=1$, where the continuous contribution also vanishes. In the limit $v\to(1/\sqrt{3})^+$ (e.g. for weak shocks), the intersection approaches the origin and therefore lies entirely within the hydrodynamic regime. In the frame comoving with the interface, the corresponding tail decays over the characteristic length
\begin{equation}
L'=\dfrac{2\gamma \eta}{3(\varepsilon{+}P)(v-c_s)} +\mathcal{O}(1)\, .
\end{equation}

For $v<0$ (wakes), the left-moving sound branch intersects the line $i\omega=vik$ only for $-1/\sqrt{3}<v<-1/3$. At $v=-(1/3)^-$, the intersection occurs precisely at the point $(i\omega\tau,ik\tau)=(1/4,-3/4)$, where the sound branch is absorbed into the continuous spectrum, producing a tail that decays over the characteristic length $L'=4\gamma\tau/3$. As $v$ approaches $-1/\sqrt{3}$ from above (e.g. for weak shocks), the intersection enters the hydrodynamic regime and the corresponding tail becomes macroscopic, with
\begin{equation}
L'=\dfrac{2\gamma \eta}{3(\varepsilon{+}P)(v+c_s)} +\mathcal{O}(1) \, .
\end{equation}

\section{An exactly solvable kinetic theory}

We conclude the article by presenting an exactly solvable kinetic theory with $\dim\mathcal{H}=\infty$, for which the propagator can be computed analytically at arbitrary interface velocity.

\subsection{Model assumptions}

We consider a non-degenerate gas of massless particles moving in one spatial dimension. Each particle has momentum $p\in\mathbb{R}$, energy $|p|$, and velocity $\Omega=\mathrm{sign}(p)$. The linearized phase-space distribution function is therefore $\delta f(t,x,p)$. The particles propagate through a medium held at inverse temperature $\beta$, which acts as a thermal bath and scatters them within a mean free time $\tau$. Assuming that the scatterings are fully randomizing, the Boltzmann equation takes the RTA-like form
\begin{equation}\label{boltzmannscatteringparticle}
(\partial_t +\Omega\partial_x)\delta f
=
\dfrac{1}{\tau}
\left[
\dfrac{\int dp'\,\delta f(p')}
{\int dp'\,e^{-\beta|p'|}}
e^{-\beta|p|}
-\delta f
\right] \, ,
\end{equation}
where the form of the gain term is completely constrained by two requirements: that the collision integral vanishes in local equilibrium, and that the particle number is conserved under scatterings.

The quadratic free-energy density associated with the perturbation $\delta f$ is
\begin{equation}\label{freenergyRTA1D}
2\Delta\mathcal F
=
\int_{-\infty}^\infty
\dfrac{dp}{2\pi} \,
\dfrac{(\delta f)^2}
{\beta e^{-\beta|p|}}.
\end{equation}
Hence, defining $\Psi=\delta f/\sqrt{2\pi \beta e^{-\beta|p|}}$, the Onsager inner product becomes $(\Psi,\Phi)=\int_{-\infty}^{\infty}\Psi^*(p)\Phi(p)\,dp $,
so that the relevant Hilbert space is $\mathcal H=L^2(\mathbb R)$. The Boltzmann equation \eqref{boltzmannscatteringparticle} then becomes
\begin{equation}\label{boltzmannscatteringparticle2}
(\partial_t +\Omega\partial_x)\Psi
=
\dfrac{1}{\tau}
\left[
(e_0,\Psi)e_0
-\Psi
\right],
\qquad
e_0
=
\dfrac{\sqrt{e^{-\beta|p|}}}
{\left\|\sqrt{e^{-\beta|p|}}\right\|}\, .
\end{equation}
We now exploit the simplicity of this model to compute the propagator analytically for an interface moving at arbitrary velocity.

\subsection{Block decomposition of the dynamics}

Define $e_1=\Omega e_0$, which is also normalized. We introduce the orthogonal decomposition $\Psi=\Psi_0 e_0+\Psi_1 e_1+\Psi_+^\perp+\Psi_-^\perp$, where $\Psi_\pm^\perp$ is a function in $L^2(\mathbb{R})$ that is supported only on positive (resp. negative) $p$, and is orthogonal to $e_0$ and $e_1$. The resulting decomposition of the Hilbert space, $\mathcal{H}=\text{span} \{e_0, e_1\}\oplus \{ \Psi{\in} L^2(\mathbb{R}^+)|\Psi{\perp} e_{0,1}\}\oplus  \{ \Psi{\in} L^2(\mathbb{R}^-)|\Psi{\perp} e_{0,1}\}$, allows us to express $\s$ and $\E$ in block form:
\begin{equation}
\s=
\left[
\begin{array}{cc|c|c}
0 & 0 & 0 & 0 \\
0 & \tau^{-1} & 0 & 0 \\ \hline
0 & 0 & \tau^{-1}\mathbb{I} & 0 \\ \hline
0 & 0 & 0 & \tau^{-1}\mathbb{I}
\end{array}
\right] \, , \qquad \qquad \E=
\left[
\begin{array}{cc|c|c}
0 & 1 & 0 & 0 \\
1 & 0 & 0 & 0 \\ \hline
0 & 0 & \mathbb{I} & 0 \\ \hline
0 & 0 & 0 & -\mathbb{I}
\end{array}
\right] \, ,
\end{equation}
where $\mathbb{I}$ is the identity in the appropriate space. We note that the first $2\times 2$ block is exactly the Cattaneo theory studied in section \ref{cattaneoHeatSection}, with $w{=}1$. Hence, the spectrum is a double hyperbola similar to that in figure \ref{fig:CattaneoPlot} (but with asymptotes of slope 1), together with the two lines $i\omega=\tau^{-1}\pm ik$, which are associated with the remaining two sectors. The propagators can therefore be evaluated rather straightforwardly.

\subsection{Explicit determination of the propagator}

Since $w=1$, there are only three relevant cases.

For $v=0$, we have
\begin{equation}
\K(x)=
\left[
\begin{array}{cc|c|c}
1 & -\frac{x}{\tau} & 0 & 0 \\
0 & 1 & 0 & 0 \\ \hline
0 & 0 & e^{-x/\tau}\mathbb{I} & 0 \\ \hline
0 & 0 & 0 & 0
\end{array}\right] \, ,
\end{equation}
which tells us that the particle density $(e_0,\Psi)$ and the particle flux $(e_1,\Psi)$ follow the same profile as in the standard planar diffusion boundary layer: the flux is constant, and the density grows linearly. By contrast, the right-moving part of the microscopic distribution function orthogonal to $e_0$ and $e_1$ relaxes exponentially to zero over a lengthscale $\tau$, while the left-moving part is constrained to vanish.

For $v>0$, the propagator reads
\begin{equation}\label{Kkineticexactblockbowwave}
\K(x-vt)=
\left[
\begin{array}{ccc|c|c}
1 & & \frac{1}{v}\left[e^{-\frac{v(x-vt)}{\tau(1-v^2)}}{-}1 \right] & 0 & 0 \\
0 & & e^{-\frac{v(x-vt)}{\tau(1-v^2)}} & 0 & 0 \\ \hline
0 & & 0 & e^{-\frac{x-vt}{\tau(1-v)}}\mathbb{I} & 0 \\ \hline
0 & & 0 & 0 & 0
\end{array}\right] \, .
\end{equation}
This tells us that there is a hydrodynamic mode, which decays in the comoving frame over the lengthscale $L'=\tau/(\gamma v)$. This lengthscale is macroscopic at small $v$, but becomes microscopic as $v\to1$. The right-moving non-conserved degrees of freedom instead decay over the comoving lengthscale $L'=\gamma(1-v)\tau$, which is always microscopic for positive $v$. As before, the left-moving non-conserved degrees of freedom are constrained to vanish.

Finally, for $v<0$, we have
\begin{equation}
\K(x-vt)=
\left[
\begin{array}{cc|c|c}
1 & -\frac{1}{v} & 0 & 0 \\
0 & 0 & 0 & 0 \\ \hline
0 & 0 & e^{-\frac{x-vt}{\tau(1-v)}}\mathbb{I} & 0 \\ \hline
0 & 0 & 0 & 0
\end{array}\right] \, .
\end{equation}
Hence, there is no hydrodynamic wake, and the flux $(e_1,\Psi)$ is constrained to vanish. There is, however, a non-hydrodynamic tail, which decays over the comoving lengthscale $L'=\gamma(1-v)\tau$. As $v\to-1$, this wake becomes macroscopic in the comoving frame.

\subsection{Application to an explicit setup}

As an illustration, consider an external source traveling at a velocity $v>0$ across the system and releasing along the line $x=vt$ a beam of forward-moving particles that are not thermally distributed. This is a bow-wave problem, and the solution takes the form $\Psi(x-vt)=\K(x-vt)\Phi+Ae_0$,
where $\K$ is given by \eqref{Kkineticexactblockbowwave}, and $\Phi(p)$ is the momentum distribution function of the emitted beam. The counter-term $Ae_0$ simply redefines the background equilibrium state, and one may choose $A$ so that $\Psi(+\infty)=0$. Since the emitted particles all have positive momentum, one has $(e_0,\Phi)=(e_1,\Phi)$, so that
\begin{equation}\label{exactsolutonbowwaverta}
\Psi(\ell)
=
\frac{1}{v}e^{-\gamma v\ell}(e_0,\Phi)(e_0+ve_1)
+
e^{-\gamma(1+v)\ell}
\left[
\Phi-(e_0,\Phi)(e_0+e_1)
\right] \, ,
\end{equation}
where we have introduced the dimensionless comoving distance $\ell=x'/\tau=\gamma(x-vt)/\tau$.

The first term is the hydrodynamic bow wave, which decays over the macroscopic lengthscale $L'=\tau/(\gamma v)$. The second term is the microscopic nonequilibrium component of the emitted beam, which relaxes over the lengthscale $L'=\tau/[\gamma(1+v)]$. The latter is much shorter than the hydrodynamic lengthscale only for small $v$. Figure \ref{fig:BowWaveExample} illustrates the resulting profile.

\begin{figure}[h!]
    \centering
\includegraphics[width=0.4\linewidth]{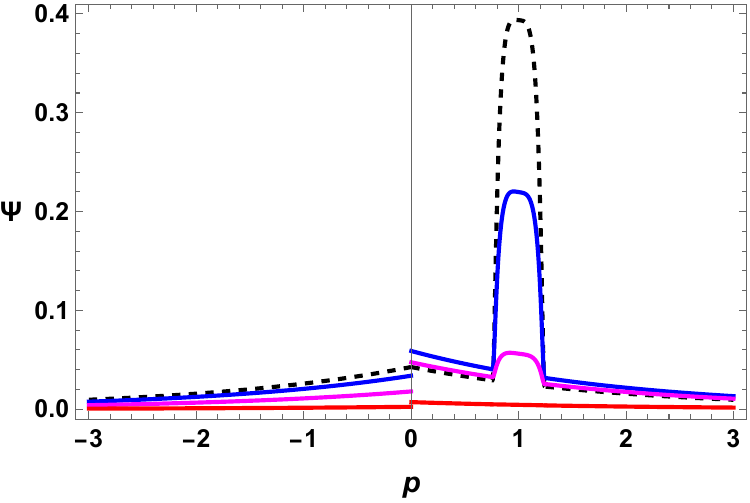}
\hspace{0.5cm}
\includegraphics[width=0.425\linewidth]{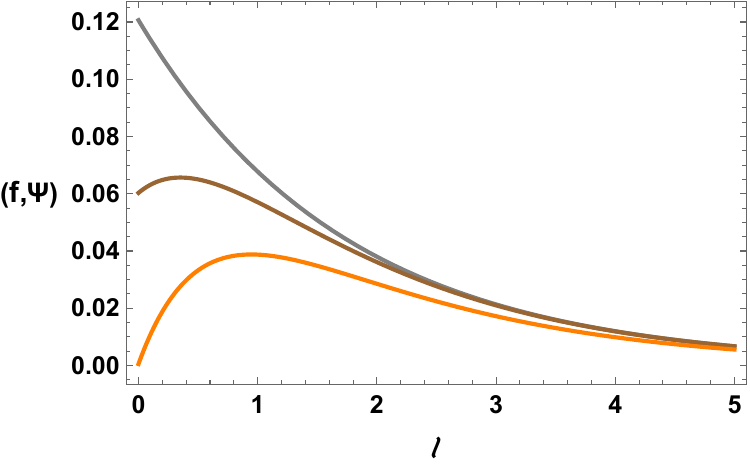}
\caption{Bow wave in a one-dimensional ultrarelativistic gas of probe particles in a Markovian bath with $\beta=1$, modeled within RTA. \textit{Left panel}: Distribution function \eqref{exactsolutonbowwaverta} for $v=1/2$, sampled at the comoving distances $\ell=0$ (dashed), $0.4$ (blue), $1.5$ (magenta), and $5$ (red). The emitted beam is described by a bump function $\Phi$ centered at $\beta p=1$. The profile at $\ell=0$ also includes a term proportional to $e_0$, corresponding to a shift in the definition of the background equilibrium state chosen so that $\Psi(+\infty)=0$. This choice ensures that, in the frame comoving with the interface, the perturbation carries no net particle flux, i.e. $\delta J^{x'}=\gamma(e_1-ve_0,\Psi)=0$. \textit{Right panel}: Profiles of selected moments $(f,\Psi)$ of the distribution function, with $f=e_0$ (gray), $(\beta^2p^2-1)e_0$ (brown), and $\frac{\beta p}{3}(\beta^2p^2-1)e_0$ (orange). The particle density (gray) is purely hydrodynamic, and therefore decays as a single exponential. By contrast, the other two moments receive both hydrodynamic and microscopic contributions, the latter decaying over the shorter kinetic lengthscale. Note that, for $v=1/2$, the hydrodynamic and microscopic lengthscales are comparable.}
    \label{fig:BowWaveExample}
\end{figure}

\section{Conclusions}

We have developed a general framework for studying linear disturbances localized near moving interfaces in relativistic media. The central result is a propagator representation of the interface solution, obtained by Laplace transformation of the equations of motion. This propagator is expressed as a contour integral enclosing the non-positive singularities of the resolvent. Under suitable spectral assumptions, its matrix elements admit a spectral representation in terms of the evanescent excitations of the theory, namely plane waves with imaginary wave number.

This formulation leads to a simple geometric interpretation of interface-localized solutions. For a given interface velocity $v$, the admissible modes are obtained by intersecting the spectrum of the homogeneous medium with the line $i\omega=vik$ in the $\{i\omega,ik\}$ plane, where both coordinates are real (as discussed in \cite{GavassinoLorentzianRelxation:2026seq}). The characteristic decay lengths of the interface solution are in one-to-one correspondence with the $ik$ coordinates of these intersections. The interface problem is thereby reduced to a geometric construction on the relaxation plane, independent of the particular realization of the microscopic dynamics.

We illustrated the formalism with applications ranging from transient hydrodynamics to kinetic theory. In Cattaneo theory, the geometric construction shows that bow waves are governed by hydrodynamic diffusion, whereas hydrodynamic wakes do not exist. In Israel--Stewart theory, hydrodynamic shock tails appear only beyond the sound speed in bow waves, and only below the sound speed in wakes. Moreover, exponentially localized bow waves disappear above the characteristic velocity, implying that Israel--Stewart theory cannot smooth arbitrarily strong shocks (in agreement with \cite{OlsonShock1990}) because its spectrum remains bounded by the characteristic cone. By contrast, the continuous spectrum of kinetic theory extends all the way to the light cone, allowing interface structures of arbitrarily small thickness.

Beyond the specific examples considered here, the construction applies to any linear theory admitting the canonical Onsager form \eqref{Boltzmann} assumed in this work. We therefore expect it to provide a useful framework for analyzing boundary layers, wakes, shock-wave tails, and related interface phenomena across a broad range of relativistic transport theories.

\section*{Acknowledgements}

LG is supported by a MERAC Foundation prize grant,  an Isaac Newton Trust Grant, and funding from the Cambridge Centre for Theoretical Cosmology.

\appendix
\section{Spectral representation of the propagator in infinite dimensions}\label{AAA}

Consider again equation \eqref{Boltzmann}, where $\Psi$ now takes values in an infinite-dimensional Hilbert space $\mathcal{H}$. The operators $\s$ and $\E$ are still self-adjoint, and $\E$ is bounded by causality, in the sense that $||\E||\leq 1$. The relaxation operator $\s$ is still non-negative definite, and is in general unbounded. However, we will also make the additional assumption (which is trivially true in finite dimensions) that the non-hydrodynamic spectrum is gapped, namely that there is a finite positive number $\tau$ (the largest relaxation time) such that $\text{Spectrum}(\s)\subseteq \{0\}\cup [1/\tau,+\infty]$, where $0$ is an eigenvalue of finite multiplicity.

The spectrum of the theory is defined as the set of complex couples $(i\omega,ik)$ such that $\s+ik\E-i\omega$ does not have a bounded inverse. Fixed $v\in [-1,1]$, we then have that the couple $(i\omega,ik)=(sv,s)$ belongs to the spectrum if and only if the operator $\G(s)=\s+(\E-v)s$ does not admit a bounded inverse. We then have the following theorem, which is the infinite-dimensional generalization of Theorem \ref{theoInvert}.

\begin{theorem}\label{theo1infity}
Suppose that there exists no non-zero vector $\Phi\in \mathcal{H}$ such that $\s\Phi=(\E-v)\Phi=0$. Then, $\G(s)^{-1}$ is bounded-holomorphic throughout the complex $s$-plane, except on a subset of the real axis.
\end{theorem}

\begin{proof}
The family $\G(s)$ is holomorphic of type (A) in the sense of Kato \cite[\S 7.1.2, Pr.~1.2]{Kato_Perturbation_Theory}. By \cite[\S 7.1.2, Th.~1.3]{Kato_Perturbation_Theory}, whenever $0$ lies in the resolvent set of $\G(s)$, the inverse $\G(s)^{-1}$ exists and is bounded-holomorphic in a neighborhood of $s$. It therefore suffices to show that $s\in\mathbb C\setminus\mathbb R$ implies $0\in\text{Resolvent-set}(\G(s))$.

Fix $s\in\mathbb C\setminus\mathbb R$. We first show that $\G(s)$ is bounded below, namely that there exists a constant $C>0$ such that $||\G(s)\Psi||\ge C||\Psi||$. Suppose by contradiction that there exists a sequence of unit vectors $\Psi_n$ such that $||\G(s)\Psi_n||\to0$. By the Cauchy--Schwarz inequality, $(\Psi_n,\G(s)\Psi_n)\to0$. Taking real and imaginary parts, and using $\mathrm{Im}(s)\neq0$, we obtain $(\Psi_n,(\E-v)\Psi_n)\to0$ and $(\Psi_n,\s\Psi_n)\to0$. Let $\mathbb P$ be the projector onto $\ker(\s)$, and $\mathbb Q=1-\mathbb P$. Since $\s\ge\mathbb Q/\tau$, the condition $(\Psi_n,\s\Psi_n)\to0$ implies that $\mathbb Q\Psi_n\to0$. Moreover, $\mathbb P\Psi_n$ is a sequence of vectors in the sphere $\ker(\s)\cap\{||\Psi||\le1\}$, which is compact because $\ker(\s)$ has finite dimension. Thus, $\Psi_n$ admits a convergent subsequence with limit $\Psi\in\ker(\s)$ and $||\Psi||=1$. Along this subsequence, we have $\s\Psi_n=\G(s)\Psi_n-s(\E-v)\Psi_n\to-s(\E-v)\Psi$, because $\G(s)\Psi_n\to0$ and $\E$ is bounded. Since $\s$ is self-adjoint, it is closed. Hence, $\Psi$ lies in the domain of $\s$, and $\s\Psi=-s(\E-v)\Psi$. But $\Psi\in\ker(\s)$, so $\s\Psi=(\E-v)\Psi=0$, contradicting our assumption.

It remains to prove that $\G(s)$ is surjective. Since $\G(s)^\dagger=\s+s^*(\E-v)$ and $s^*\notin\mathbb R$, the same argument applies to $\G(s)^\dagger$, implying that $\ker\G(s)^\dagger=\{0\}$. Therefore, $\overline{\mathrm{ran}\,\G(s)}=[\ker\G(s)^\dagger]^\perp=\mathcal H$. Finally, $\G(s)$ is closed because it is the sum of the closed operator $\s$ and the bounded operator $s(\E-v)$. Since $\G(s)$ is bounded below, its range is closed. Consequently, $\operatorname{ran}\G(s)$ is both closed and dense in $\mathcal H$, and therefore $\operatorname{ran}\G(s)=\mathcal H$. 

Thus, $\G(s)$ is bijective and $\G(s)^{-1}$ is bounded.
\end{proof}

We also have a second useful result.

\begin{theorem}\label{theo2infity}
Under the assumptions of Theorem~\ref{theo1infity}, the family $\G(s)^{-1}$ is bounded-holomorphic in a region $0<|s|<s_m$, for some $s_m>0$.
\end{theorem}

\begin{proof}
Since $\mathcal{H}=\ker(\s)\oplus[\ker(\s)]^\perp$, we have the block decomposition
\begin{equation}
\G(s)=
\begin{bmatrix}
s\mathbb{P}(\E-v)\mathbb{P} & s\mathbb{P}(\E-v)\mathbb{Q} \\
s\mathbb{Q}(\E-v)\mathbb{P} & \s+s\mathbb{Q}(\E-v)\mathbb{Q}
\end{bmatrix}
\equiv
\begin{bmatrix}
s\mathbb{A} & s\mathbb{B}^\dagger \\
s\mathbb{B} & \mathbb{D}(s)
\end{bmatrix}\, ,
\end{equation}
where each block acts between the corresponding Hilbert subspaces. On $[\ker(\s)]^\perp$, the restriction of $\s$ is invertible and satisfies $||\s^{-1}||\leq\tau$. Moreover, $||\mathbb{Q}(\E-v)\mathbb{Q}||\leq2$. Therefore, $\mathbb{D}(s)
=
\s\left[1+s\s^{-1}\mathbb{Q}(\E-v)\mathbb{Q}\right]$
is invertible whenever $2\tau|s|<1$, with inverse given by a convergent Neumann series. In particular, $\mathbb{D}(s)^{-1}$ is bounded-holomorphic for $|s|<(2\tau)^{-1}$.
Within this disk, $\G(s)$ is invertible if and only if its Schur complement $\mathbb{F}(s)
=
s\mathbb{A}
-s^2\mathbb{B}^\dagger\mathbb{D}(s)^{-1}\mathbb{B}$
is invertible. Since $\ker(\s)$ is finite-dimensional, this is equivalent to $\det\mathbb{F}(s)\neq0$. The function $\det\mathbb{F}(s)$ is analytic for $|s|<(2\tau)^{-1}$ and vanishes at $s=0$. It cannot vanish identically, because Theorem~\ref{theo1infity} implies that $\G(s)$, and hence $\mathbb{F}(s)$, is invertible whenever $s\notin\mathbb{R}$. Thus, $s=0$ is an isolated zero of $\det\mathbb{F}(s)$. Consequently, there exists $s_m>0$ such that $\G(s)$ is invertible for $0<|s|<s_m$.

Finally, by \cite[\S 7.1.2, Th.~1.3]{Kato_Perturbation_Theory}, $\G(s)^{-1}$ is bounded-holomorphic throughout this punctured disk.
\end{proof}

The above results tell us that we can still define a family of operators of the form \eqref{Kfund}, with the integral over a contour $\Gamma$ as in figure \ref{fig:CloseTheLoop}. Such a contour encloses the origin and no positive singularity, as we can place $\Gamma$ on the left of $s_m$. However, there is now an important difference with respect to the finite-dimensional case, namely that now the singular set of $\G(s)^{-1}$ may extend to $-\infty$, making $\Gamma$ non-compact, and possibly causing convergence issues. In practice, we do not expect this to be a problem because, for $x-vt>0$, the exponential factor $e^{s(x-vt)}$ is likely to suppress the integrand as $s\to -\infty$. Still, to simplify the proof, we shall assume in the theorem below that the singular set of $\G(s)^{-1}$ is bounded.

\begin{theorem}\label{theo3infty}
Suppose that the assumptions of Theorem \ref{theo1infity} hold, and that there exists a compact curve $\Gamma$ that encloses all non-positive part of the spectrum. Then, the family of operators
\begin{equation}
\mathbb{K}(x-vt)=\oint_\Gamma\frac{ds}{2\pi i}\,
e^{s(x-vt)}\,\G(s)^{-1}(\E-v)
\end{equation}
have the following properties:
\begin{itemize}
\item[\textup{(i)}] They are bounded operators, and depend holomorphically on $x-vt$;
\item[\textup{(ii)}] They solve equation \eqref{Boltzmann}, in the sense that $(\partial_t+\s+\E\partial_x)\mathbb{K}=0$;
\item[\textup{(iii)}] The operator $\mathbb{K}(0)$ is a projector, i.e. $\K(0)^2=\K(0)$. It acts as the identity on any solution of the form $\Psi=\Psi(0)e^{s(x-vt)}$ with $s$ non-positive, and annihilates any solution of the same form with $s$ positive.
\end{itemize}
\end{theorem}

\begin{proof}
\textup{(i)} Since $\Gamma$ is compact and does not intersect the singular set of $\G(s)^{-1}$, the operator-valued function $\G(s)^{-1}(\E{-}v)$ is bounded-holomorphic in a neighborhood of $\Gamma$. Hence, the contour integral exists in operator norm, and we can proceed as in the proof of point (i) of Theorem \ref{theoK}.

\textup{(ii)} For every $s\in\Gamma$, the operator $\G(s)^{-1}$ maps $\mathcal H$ into $\operatorname{Dom}(\s)$. Furthermore, $\s\G(s)^{-1}=1-s(\E-v)\G(s)^{-1}$, so the integrand is continuous on $\Gamma$ in the graph norm of $\s$. We may therefore differentiate under the integral sign and bring $\s$ inside the integral (the bounded operator $\E$ may also be brought inside). Using $\G(s)=\s{+}s(\E{-}v)$, we obtain
\begin{equation}
(\partial_t+\s+\E\partial_x)\mathbb{K}
=\oint_\Gamma\frac{ds}{2\pi i}\,
e^{s(x-vt)}(\E-v)\, ,
\end{equation}
which vanishes.

\textup{(iii)} The proof here is exactly the same as the proof of point (iii) of Theorem \ref{theoK}. 
\end{proof}

Since $\K(x-vt)$ are bounded-holomorphic, their domain can be extended to the whole Hilbert space, and the matrix elements $(\Phi_1,\K(x-vt)\Phi_2)$ are holomorphic functions of $x-vt$. The matrix element can be brought inside the integral by boundedness, and thus we can write 
\begin{equation}
(\Phi_1,\K(x-vt)\Phi_2)=\oint_\Gamma\frac{ds}{2\pi i}\,
e^{s(x-vt)}\,(\Phi_1,\G(s)^{-1}(\E-v)\Phi_2) \, .
\end{equation}
Now, $(\Phi_1,\G(s)^{-1}(\E-v)\Phi_2) $ is holomorphic away from the singularities of $\G(s)^{-1}$, so we can take the limit in which the upper and the lower side of the contour $\Gamma$ approach the real axis, giving
\begin{equation}
(\Phi_1,\K(x{-}vt)\Phi_2)=\int_{-\infty}^{0^+} ds\,
e^{s(x-vt)}\,\underbrace{\left(\Phi_1,\dfrac{\G(s{-}i0)^{-1}-\G(s{+}i0)^{-1}}{2\pi i}  (\E{-}v)\Phi_2\right)}_{\rho(s)} \, . 
\end{equation}
The spectral density $\rho(s)$ is in general a distribution, which vanishes wherever $\G(s)$ has a bounded inverse, and thus is necessarily supported on the spectrum of the theory. Therefore, the exponential factors appearing in the asymptotic tails of the interface solutions are completely determined by the intersection of the spectrum of the theory with the line $i\omega=vik$, exactly as in the finite-dimensional case.

\section{An invariant longitudinal sector of RTA kinetic theory}\label{BBB}

\subsection{Linearization of RTA}

We consider an ideal gas of massless bosons characterized by a kinetic distribution function $f(x^\mu,p^\alpha)$, which gives the occupation number of the single-particle state with four-momentum $p^\alpha$ at spacetime point $x^\mu$ (with $p^\alpha p_\alpha=0$). Working in the Relaxation-Time Approximation (RTA), the nonlinear equation of motion is
\begin{equation}\label{nonlinearboltzmanncattaneo}
p^\mu \partial_\mu f
=
-\dfrac{u_\mu p^\mu}{\tau}
\left[
f_{\mathrm{eq}}(-\beta_\alpha p^\alpha)-f
\right],
\end{equation}
where $f_{\mathrm{eq}}(z)=(e^z-1)^{-1}$ is the local-equilibrium occupation number.
The inverse-temperature four-vector $\beta^\alpha$ \cite{BecattiniBeta2016,GavassinoTermometri} depends on $f$ through Landau matching (and $u^\alpha\propto \beta^\alpha$, with $u^\alpha u_\alpha=-1$) \cite{OlsonLifsh1990,cercignani_book}. Indeed, the energy-momentum tensor is
\begin{equation}
T^{\mu\nu}
=
\int
\dfrac{g\,d^3p}{(2\pi)^3p^t}
\,p^\mu p^\nu f,
\qquad
(g \text{ the spin degeneracy}) \, .
\end{equation}
Hence, multiplying both sides of \eqref{nonlinearboltzmanncattaneo} by $gp^\nu/p^t$, integrating over all momenta, and imposing $\partial_\mu T^{\mu\nu}=0$, we obtain
\begin{equation}
\int
\dfrac{g\,d^3p}{(2\pi)^3p^t}
(-u_\mu p^\mu)p^\nu
\left[
f_{\mathrm{eq}}(-\beta_\alpha p^\alpha)-f
\right]
=
0  \qquad \Longleftrightarrow \qquad u_\mu (T_\text{eq}^{\mu \nu}-T^{\mu \nu})=0 \, .
\end{equation}

Linearizing around the homogeneous equilibrium state $f_{\mathrm{eq}}(p^t/T)\equiv f_{\mathrm{eq}}$ (with constant temperature $T$), and assuming that all gradients are along the $x$ direction, we obtain
\begin{equation}\label{linearRTA}
\begin{split}
(\partial_t+\Omega\partial_x)\delta f
&=
\tau^{-1}
\left[
f_{\mathrm{eq}}(1+f_{\mathrm{eq}})
p^\alpha\delta\beta_\alpha
-\delta f
\right] \, ,\\
\delta\beta_t
&=
\dfrac{\displaystyle
\int d^3p\,p^t\delta f}
{\displaystyle
\int d^3p\,f_{\mathrm{eq}}(1+f_{\mathrm{eq}})(p^t)^2} \, ,\\
\delta\beta_j
&=
\dfrac{\displaystyle
\int d^3p\,p^j\delta f}
{\displaystyle
\int d^3p\,f_{\mathrm{eq}}(1+f_{\mathrm{eq}})(p^j)^2},
\qquad
(\text{no sum over }j) \, ,
\end{split}
\end{equation}
where $\Omega=p^x/p^t$ denotes the $x$ component of the particle velocity.

Finally, the quadratic perturbation of the free-energy density is
\begin{equation}\label{freenergyRTA}
2\Delta\mathcal F
=
\int
\dfrac{g\,d^3p}{(2\pi)^3}
\dfrac{(\delta f)^2}
{f_{\mathrm{eq}}(1+f_{\mathrm{eq}})\beta} \, .
\end{equation}

\subsection{A useful simplification}

It is immediate to verify that perturbations of the form
\begin{equation}\label{ansatzRTA}
\delta f
=
\sqrt{\dfrac{2\beta^3}{c_v}}
\,f_{\mathrm{eq}}(1+f_{\mathrm{eq}})
p^t
\Psi(\Omega)
\end{equation}
constitute an invariant sector of the linearized equation of motion. Indeed, substituting \eqref{ansatzRTA} into \eqref{linearRTA}, all dependence on $p^t$ cancels, which is only possible because the relaxation time $\tau$ is independent of the particle energy. The ansatz also removes any dependence on the azimuthal angle in the transverse plane, implying $\delta\beta_y=\delta\beta_z=0$, so that all flows are longitudinal.

For real $\Psi$, inserting \eqref{ansatzRTA} into \eqref{freenergyRTA} gives $2\Delta\mathcal F=\int_{-1}^1\Psi(\Omega)^2\,d\Omega$. After complexification, the corresponding Onsager inner product is therefore $(\Psi,\Phi)=\int_{-1}^1\Psi(\Omega)^*\Phi(\Omega)\,d\Omega$, so that $\mathcal H=L^2([-1,1])$. Furthermore, substituting \eqref{ansatzRTA} into \eqref{linearRTA}, we obtain
\begin{equation}
(\partial_t+\Omega\partial_x)\Psi
=
\tau^{-1}
\left[
\dfrac{(1,\Psi)}{(1,1)}
+\Omega\dfrac{(\Omega,\Psi)}{(\Omega,\Omega)}
-\Psi
\right] \, .
\end{equation}

Thus, the velocity operator $\E$ is simply multiplication by $\Omega$, while the relaxation operator is $\s=\tau^{-1}(1-\mathbb P)$, where $\mathbb P$ is the orthogonal projector onto $\operatorname{span}\{1,\Omega\}$. In particular, $\|\E\|=1$, $\ker(\s)=\operatorname{span}\{1,\Omega\}$, and the non-hydrodynamic spectrum consists of the single eigenvalue $\tau^{-1}$. Hence, this model satisfies all the assumptions of Appendix~\ref{AAA}.

\bibliography{Biblio}

\label{lastpage}
\end{document}